\documentclass[%
 reprint,
superscriptaddress,
 amsmath,amssymb,
 aps,
 pra,
]{revtex4-2}

\usepackage[T1]{fontenc}

\usepackage{graphicx}
\usepackage{dcolumn}
\usepackage{bm}
\usepackage{comment}
\usepackage[colorlinks=true]{hyperref}
\graphicspath{{manuscript/}}

\usepackage{times}

\usepackage{mathtools}
\usepackage{amsmath}
\usepackage{amssymb}
\allowdisplaybreaks

\usepackage{amsthm}
\newtheorem{theorem}{Theorem}

\newtheorem{lemma}[theorem]{Lemma}

\newtheorem{definition}{Definition}
\newtheorem{proposition}[theorem]{Proposition}
\newtheorem*{remark}{Remark}

\newcommand{\ket}[1]{|#1\rangle}

\newcommand{\Tr}{\operatorname{Tr}}
\newcommand{\Pauli}{\mathcal{P}}

\newcommand{\E}{\mathbb E}
\newcommand{\Sym}{\operatorname{Sym}}

\newcommand{\R}{\mathbb R}

\begin{document}


\title{Enhancing Entanglement Purification with Shared Randomness}

\author{Allen Zang}
\affiliation{Pritzker School of Molecular Engineering, University of Chicago, Chicago, IL, USA}

\author{Bikun Li}
\affiliation{Pritzker School of Molecular Engineering, University of Chicago, Chicago, IL, USA}

\author{Xinan Chen}
\affiliation{Department of Electrical and Computer Engineering, University of Illinois Urbana-Champaign, Urbana, IL, USA}

\author{Eric Chitambar}
\affiliation{Department of Electrical and Computer Engineering, University of Illinois Urbana-Champaign, Urbana, IL, USA}

\author{Liang Jiang}
\affiliation{Pritzker School of Molecular Engineering, University of Chicago, Chicago, IL, USA}

\author{Martin Suchara}
\affiliation{IonQ Inc., College Park, MD, USA}

\author{Tian Zhong}
\affiliation{Pritzker School of Molecular Engineering, University of Chicago, Chicago, IL, USA}

\date{\today}

\begin{abstract}
Entanglement purification protocols (EPPs) are essential for improving entanglement fidelity to support fault-tolerant distributed quantum information processing. Practical entanglement sources are often heterogeneous and source labels may be unavailable at the EPP layer. We show that classical shared randomness, together with buffer memories, can enhance entanglement purification when source labels are unavailable, without state characterization or EPP circuit optimization. The strategy is to accumulate multiple entanglement distribution rounds and then use shared randomness to shuffle all the stored entangled states before packaging them as inputs to the EPP. For any $n$ Werner sources and any fixed $n$-to-1 bilocal Clifford EPP, we prove that accumulating and shuffling improves the expected success probability and the success-weighted output Bell fidelity over the baseline without accumulating and shuffling, for every $n$, for every finite number of accumulation rounds and in the asymptotic limit, and the improvement increases monotonically with the number of accumulation rounds. 
\end{abstract}

\maketitle

\textit{Introduction}---
Entanglement distribution quantum networks~\cite{kimble2008quantum,wehner2018quantum} are long envisioned as the backbone of future scalable distributed quantum information processing~\cite{knorzer2026distributed}. While quantum networks have multiple architectures that use technologies of different levels of realization difficulty~\cite{muralidharan2016optimal,azuma2023quantum}, entanglement purification protocols (EPPs)~\cite{bennett1996purification,deutsch1996quantum,bennett1996mixed,dur2007entanglement} are necessary for suppressing errors in quantum networks in the near term, and they will remain useful in further future as pre-processing before quantum error correction~\cite{gottesman1997stabilizer}. In realistic quantum networks, a persistent complication is that the input pairs to an EPP are almost never identical~\cite{zhou2020purification,qin2023efficient,zang2025no,zhou2025observation,zang2025entanglement,fayyaz2026purify}, due to the probabilistic nature of remote entanglement generation, inevitable decoherence of quantum memories which store entangled states, and other features of quantum network protocols~\cite{dur1999quantum,dur2003entanglement,chakraborty2019distributed,kolar2022adaptive,ghaderibaneh2022pre,zang2023entanglement,inesta2023performance,davies2024entanglement,zang2024analytical,zhan2025design}. Moreover, non-identical inputs are not merely a modeling nuisance, as recent results have shown fundamental limitations on the capability of fixed EPPs for general heterogeneous inputs~\cite{zang2025no}. 

The evidence of potential undesirable performance of EPPs for practical non-identical entangled states naturally raises a question: If two parties receive entangled pairs from heterogeneous sources and cannot reliably identify which source produced each pair, how can we improve the performance an entanglement purification protocol? Some straightforward approaches include EPP optimization~\cite{rozpkedek2018optimizing,krastanov2019optimized} which may require quantum benchmarking~\cite{martinis2015qubit,harper2020efficient,eisert2020quantum,kliesch2021theory,hashim2025practical} to obtain information of the input states, and careful scheduling of EPPs based on expected output, which also requires the knowledge of entangled states to purify. In this work, we show that shared randomness, a \textit{classical resource} which can be pre-generated locally and broadcast over classical communication channels, can enhance the performance of a fixed EPP in a surprisingly simple and inexpensive way, through symmetrization of the input state packages seen by the EPP.

The strategy, which we call accumulating and shuffling (AS), uses commonly assumed buffer quantum memories~\cite{wu2023qucomm,liu2025hardware} to accumulate entangled states distributed in multiple rounds and shared randomness to apply the same random permutation to Alice's and Bob's local memories before forming EPP input packages. Shared randomness is essential because Alice and Bob must choose matching local halves of the same entangled pairs. The result is an input-independent and resource-efficient randomized packaging primitive whose enhancement is provable, without the need for optimizing the EPP circuit or characterizing the states being distributed. 

In the following, we first formulate the operational model of the entanglement distribution scenario and explain the effect of the AS strategy therein. Then we demonstrate the mechanism of AS with a simplest scenario with two Werner souces and the bilocal CNOT EPP~\cite{bennett1996purification,deutsch1996quantum}. Finally we generalize the provable enhancement by AS to any $n$ Werner sources and any fixed $n$-to-1 bilocal Clifford entanglement purification protocol~\cite{jansen2022enumerating}. 

\begin{figure*}[t]
    \centering
    \includegraphics[width=0.9\linewidth]{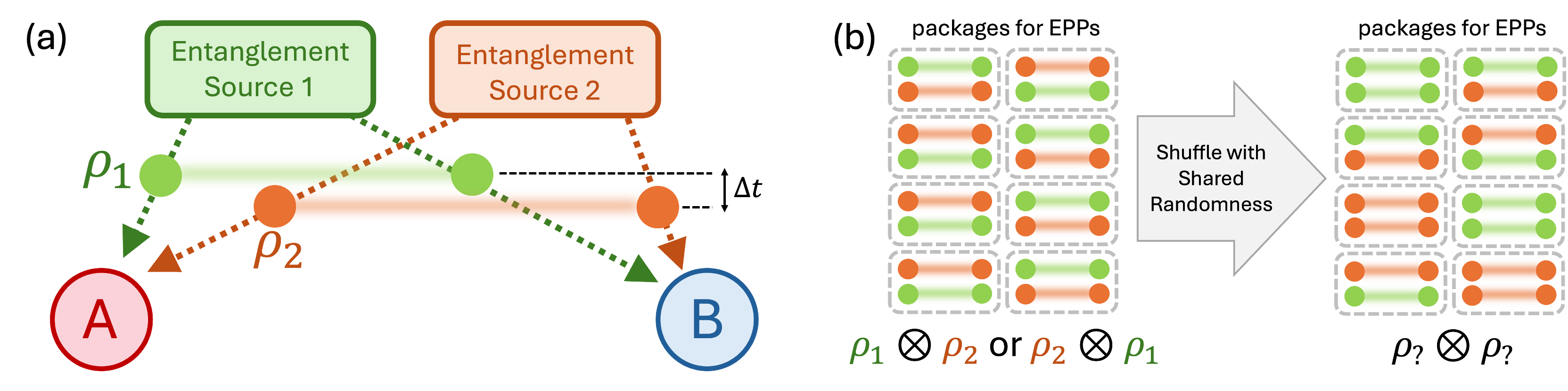}
    \caption{Two-source schematic of accumulating and shuffling. (a) Alice and Bob receive one pair from each of two non-identical sources in each distribution round, but the source labels of the received pairs are not available to the purification layer, likely due to e.g. a time jitter $\Delta t$ between the distribution of the entangled states that cannot be known beforehand. (b) With shared randomness and buffer memories, they accumulate over $m$ rounds, apply the same random shuffle to their local memories, and group the stored pairs into input packages for the EPP. The mechanism applies to general $n$ sources and $n$-to-1 EPPs.}
    \label{fig:schematics}
\end{figure*}

\textit{Operational model}---
Consider $n$ entanglement sources that, in one distribution round, provide Alice and Bob with noisy bipartite states $\rho_1,\dots,\rho_n$. The entanglement purification layer is ignorant of the identities of the sources: Both parties receive the $n$ pairs but do not know their source labels. Thus the natural non-AS baseline is not an omniscient ordering of the source states or a best-pair comparator. Instead, it should be the uniformly random source-label permutation
\begin{align}\label{eq:noas-baseline}
    \rho_{\mathrm{noAS}} = \frac{1}{n!} \sum_{\pi\in S_n} \rho_{\pi(1)}\otimes\cdots\otimes\rho_{\pi(n)},
\end{align}
where $S_n$ is the symmetric group of order $n$ and $\pi$ denotes a permutation in this group. We elaborate on the considerations and justifications of the operational assumptions in the End Matter.

With AS, Alice and Bob accumulate entangled states distributed over $m$ rounds, so that they hold $m$ copies of each source state $\rho_i$. They then use shared randomness to uniformly shuffle the full collection of $mn$ entangled pairs and divide the shuffled list into $m$ packages of size $n$. Each package is used once as an EPP input, so that the number of EPP attempts is fixed for fair comparison with the non-AS scenario, where Alice and Bob also perform $m$ times of EPP after $m$ entanglement distribution rounds.

Let $U_\pi$ denote the unitary that permutes the $n$ tensor factors of one package according to permutation $\pi\in S_n$. We can define the uniform-randomly permuted (symmetrized) $n$-partite state as $\Sym(X) = \frac{1}{n!}\sum_{\pi\in S_n}U_\pi X U_\pi^\dagger$. The effective single-package state produced by AS has the following closed form.
\begin{theorem}[Effective package state after AS]\label{thm:rho_eff}
    After accumulating $m$ copies of each $\rho_i$, uniformly shuffling all $mn$ pairs, and forming packages of size $n$, the expected state of any one package is
    \begin{align}\label{eq:rho-eff-general-main}
        \rho_{\mathrm{eff}}^{(m)} = \sum_{\substack{k_1+\cdots+k_n=n\\0\leq k_i\leq m}} \frac{\prod_{i=1}^n \binom{m}{k_i}}{\binom{mn}{n}} \Sym\left(\rho_1^{\otimes  k_1}\otimes\cdots\otimes\rho_n^{\otimes k_n}\right).
    \end{align}
    In the limit $m\to\infty$, this converges to $\rho_{\mathrm{eff}}^{(\infty)} = \bar{\rho}^{\otimes n}$ where $\bar{\rho}=\frac{1}{n}\sum_{i=1}^n\rho_i$.
\end{theorem}

The proof is based on sampling without replacement. For one package, we can use an occupation vector $k=(k_1,\dots,k_n)$ to record how many slots are filled by copies of each source state: $k_i$ is the number of copies of $\rho_i$ in that package. Thus $k_i\geq 0$, $\sum_i k_i=n$, and $k_i\leq m$ because only $m$ copies of each source have been accumulated according to our operational assumption. Choosing $k_i$ copies from source $i$ for every $i$ gives the multivariate hypergeometric weight $\prod_i\binom{m}{k_i}/\binom{mn}{n}$, and the assumed randomness also uniformly symmetrizes the slot order. See the Supplemental Material (Sec.~\ref{sec:eff_state})~\cite{supmat} for the complete proof, together with other equivalent proofs from different perspectives, including ordered-sampling and permutation-twirling derivations.

The relevant EPP performance follows from linearity. We use the following fact throughout the work.
\begin{lemma}\label{thm:eff_output}
    Consider an input-independent protocol with completely positive trace-non-increasing success map $\mathcal{E}_{\mathrm{succ}}$. If the effective input is $\rho=\sum_\alpha q_\alpha\rho_\alpha$, then
    \begin{align}
        p_{\mathrm{succ}}(\rho) &= \sum_\alpha q_\alpha p_{\mathrm{succ}}(\rho_\alpha),\\
        \rho_{\mathrm{succ}}(\rho) &= \frac{\sum_\alpha q_\alpha p_{\mathrm{succ}}(\rho_\alpha)\rho_{\mathrm{succ}}(\rho_\alpha)} {p_{\mathrm{succ}}(\rho)},
    \end{align}
    whenever the success probability $p_{\mathrm{succ}}(\rho)>0$.
\end{lemma}

\textit{Two-source example}---
We first show the mechanism in the simplest and most canonical setting with two Werner sources and the standard 2-to-1 bilocal CNOT protocol, which is equivalent to the DEJMPS protocol~\cite{deutsch1996quantum} up to local Clifford relabeling. The specific EPP is considered not only because it is the most experimentally friendly~\cite{pan2001entanglement,pan2003experimental,hu2021long,ecker2021experimental,reichle2006experimental,kalb2017entanglement,yan2022entanglement}, but also because there has been strong evidence of its optimality~\cite{rozpkedek2018optimizing,preti2022optimal,jansen2022enumerating}. A Werner state with Bell fidelity $F$ is $\rho^{(W)}(F) = F\Phi^+ + (1-F)\left(\Phi^-+\Psi^++\Psi^-\right)/3$, which can also be parameterized by visibility (Werner parameter) $w=(4F-1)/3$ as $\rho^{(W)}(w) = w\Phi^+ + (1-w)I/4$. Recall the standard Bell states $\ket{\Phi^\pm}=(\ket{00}\pm\ket{11})/\sqrt{2}$, $\ket{\Psi^\pm}=(\ket{01}\pm\ket{10})/\sqrt{2}$. We fix $\ket{\Phi^+}$ as the target Bell state, and $\Phi^\pm,\Psi^\pm$ are the projectors for the four Bell states.

For $n=2$, Theorem~\ref{thm:rho_eff} gives $\rho_{\mathrm{eff},n=2}^{(m)} = \frac{2m-2}{2m-1}\bar{\rho}\otimes\bar{\rho} + \frac{1}{2m-1}\Sym(\rho_1\otimes\rho_2)$, where $\bar{\rho}=(\rho_1+\rho_2)/2$. For Bell-diagonal inputs, the success probability and success-weighted output Bell fidelity of the bilocal CNOT protocol are invariant under exchanging the two input states. Consequently, for the Werner inputs considered here, $\Sym(\rho_1\otimes\rho_2)$ is operationally equivalent to $\rho_1\otimes\rho_2$ for these figures of merit. For this scenario we have the following result.
\begin{proposition}\label{thm:succ_improvement}
    For two Werner inputs with fidelities $F_1,F_2\in[1/2,1]$, AS improves both the DEJMPS success probability and the normalized successful output Bell fidelity over the non-AS baseline:
    \begin{align}
        p_{\mathrm{succ}}(\rho_{\mathrm{eff},n=2}^{(m)}) &\geq p_{\mathrm{succ}}(\rho_1\otimes\rho_2),\\
        F_{\mathrm{succ}}(\rho_{\mathrm{eff},n=2}^{(m)}) &\geq F_{\mathrm{succ}}(\rho_1\otimes\rho_2).
    \end{align}
    Moreover, we have monotonicity with the number of accumulation rounds $p_{\mathrm{succ}}(\rho_{\mathrm{eff},n=2}^{(m+1)}) \ge p_{\mathrm{succ}}(\rho_{\mathrm{eff},n=2}^{(m)})$ and $F_{\mathrm{succ}}(\rho_{\mathrm{eff},n=2}^{(m+1)}) \ge F_{\mathrm{succ}}(\rho_{\mathrm{eff},n=2}^{(m)})$.
\end{proposition}
The results are already clear by examining the asymptotic limit with the effective state of an input package being $\bar{\rho}\otimes\bar{\rho}$. Given Werner inputs we have the following hold for success probability difference $ p_{\mathrm{succ}}(\bar{\rho}\otimes\bar{\rho}) - p_{\mathrm{succ}}(\rho_1\otimes\rho_2) = 2(F_1-F_2)^2/9\geq 0$, and the successful fidelity difference $F_{\mathrm{succ}}(\bar{\rho}\otimes\bar{\rho})-F_{\mathrm{succ}}(\rho_1\otimes\rho_2)$ is also non-negative on $[1/2,1]^2$. Moreover, these two differences are \textit{strictly positive} except for the identical-input case ($\rho_1=\rho_2$). Quantitatively, the highest enhancement appears in the extreme case of $F_1=1/2$ and $F_2=1$, with $p_{\mathrm{succ}}(\bar{\rho}\otimes\bar{\rho})-p_{\mathrm{succ}}(\rho_1\otimes\rho_2) = 1/18$  which is a $8.3\%$ enhancement over the non-AS baseline, and $F_{\mathrm{succ}}(\bar{\rho}\otimes\bar{\rho})-F_{\mathrm{succ}}(\rho_1\otimes\rho_2) = 1/26$ which is a $5.1\%$ enhancement over the non-AS baseline. The detailed derivation can be found in the Supplemental Material (Sec.~\ref{sec:bilocal_cnot_epp_results})~\cite{supmat}. Finite $m$ enhancement immediately follows from Lemma~\ref{thm:eff_output}, because $\rho_{\mathrm{eff},n=2}^{(m)}$ is a convex mixture of the baseline package and the symmetrized package $\bar{\rho}\otimes\bar{\rho}$. Thus AS helps by partially replacing a heterogeneous input pair with the average-input pair that the EPP processes more favorably. In the following, we will see that the ``preference'' for more identical input states is in fact a very general phenomenon among EPPs, which allows us to prove the guaranteed enhancement via the AS strategy that is independent of the $n$ Werner sources, for all $n$-to-1 bilocal Clifford EPPs (biCEPs).

\textit{Mechanism and properties of biCEP}---
We now generalize from the 2-to-1 bilocal CNOT EPP to arbitrary bilocal Clifford EPPs (biCEPs).
\begin{definition}[$n$-to-1 biCEP~\cite{jansen2022enumerating}]\label{def:bicep-main}
    An $n$-to-1 bilocal Clifford entanglement purification protocol (biCEP) takes $n$ noisy Bell pairs shared by Alice and Bob. Alice applies an $n$-qubit Clifford $U$ to her halves and Bob applies the entry-wise complex conjugate $U^*$ to his halves. They then measure $n-1$ output pairs in the computational basis, compare the corresponding local outcomes by one round of classical communication, accept iff all $n-1$ measured pairs have even parity, and keep the remaining pair as the output.
\end{definition}

While the mechanism of biCEPs can be understood from different perspectives, the stabilizer-code interpretation~\cite{gottesman1997stabilizer,dur2007entanglement,goodenough2024near,shi2025stabilizer,bonilla2025constant,li2025efficient} is likely the most convenient way to analyze these protocols for this work. Let $S_0=\langle Z_2,\dots,Z_n\rangle$, i.e. a group generated by $Z_i,i=2,\dots,n$, be the stabilizer group of the trivial $[[n,1]]$ stabilizer code whose logical qubit is the first qubit, where $Z_i$ is the Pauli $Z$ operator applied on the $i$-th qubit. The Clifford $U$ transforms the trivial code to a new code $S_U = U^\dagger S_0 U = \langle g_2,\dots,g_n\rangle$ with $g_j=U^\dagger Z_j U$. Thus comparing Alice's and Bob's computational-basis measurement outcomes on qubit pair $j$ is equivalent to a bilocal stabilizer measurement of $g_{j,A}\otimes g_{j,B}^*$, where $O^*$ is the entry-wise complex conjugate of $O$.

For the tensor product of $n$ Bell-diagonal states (BDS), the joint state is equivalent to applying an $n$-qubit Pauli channel on one of the two parties that share the $n$ perfect Bell pairs $(\Phi^+)^{\otimes n}$. Without loss of generality, if a Pauli error string $E\in\mathcal P_n$ acts on Bob's side of initially perfect Bell pairs, the protocol is considered successful exactly when $E$ commutes with all elements of $S_U$, i.e. the Pauli error string belongs to the normalizer of the stabilizer group $E\in N(S_U)$, where $N(S_U)=\{h\in\mathcal{P}_n : hs=sh, \forall s\in S_U\}$ and $\mathcal{P}_n$ is the $n$-qubit Pauli group; and the accepted output is the target $\Phi^+$ exactly when the induced logical operator is trivial, i.e. $E\in S_U$. Hence a biCEP can be viewed as a process of stabilizer code-based error detection: Trivial syndrome leads to acceptance (success). In the actual biCEP circuit, we do not perform syndrome measurements with extra ancillas, but instead measure $n-1$ physical Bell pairs followed by the Clifford circuit, which can be understood as a combination of syndrome measurement and decoding. The output entangled state can then be interpreted as the ``decoded logical Bell pair''. Once accepted, errors that differ by a stabilizer will act identically on the decoded logical Bell pair, because stabilizer elements are trivial on the codespace. The output error is therefore not $E$ itself but an equivalence class in the quotient group $N(S_U)/S_U$. Explicitly, for an accepted error $E\in N(S_U)$, its logical coset is $E S_U = \{E s:s\in S_U\}$. Note that the left and right cosets are equivalent here. All Pauli operators in this set differ only by stabilizers and therefore have the same action after decoding. Modulo irrelevant global phases, the quotient $N(S_U)/S_U$ is the Pauli group of the single logical qubit of the $[[n,1]]$ code. Thus the identity coset $S_U$ gives the target output $\Phi^+$, while the other cosets correspond to logical Pauli errors on the kept pair.

For Bell diagonal state (BDS) inputs this interpretation gives closed expressions. One can always write a BDS input entangled pair as the result of applying a 1-qubit Pauli channel (statistical mixture of applying single Pauli operators) on Bob's half of $\Phi^+$, i.e. $\rho_i = \sum_{P\in\{I,X,Y,Z\}} p_i(P)(I\otimes P)\Phi^+(I\otimes P)$. For independent inputs, the joint state is $\bigotimes_{i=1}^n\rho_i = \sum_{E\in\{I,X,Y,Z\}^{\otimes n}} p(E)(I\otimes E)(\Phi^+)^{\otimes n}(I\otimes E)$. The Pauli error string $E=E_1\otimes\cdots\otimes E_n$ occurs with probability $p(E) = \prod_i p_i(E_i)$, where $E_i$ is a 1-qubit Pauli operator applied on qubit $i$. The stabilizer interpretation immediately gives
\begin{align}\label{eq:bicep-error-set-main}
    p_{\mathrm{succ}} &= \Pr(E\in N(S_U)),\\
    F_{\mathrm{succ}} &= \frac{\Pr(E\in S_U)}{\Pr(E\in N(S_U))}.
\end{align}

Equivalently, and more useful for the later AS theorem, the success probability and the target-output numerator can be written as averages of Pauli correlation, i.e. the expectation value of conjugate bilocal Pauli operators $\langle P_A\otimes P_B^*\rangle$. For an $n$-qubit Pauli operator $t$, the Bell-error state $(I\otimes E)(\Phi^+)^{\otimes n}(I\otimes E)$ is an eigenstate of $t_A\otimes t_B^*$ with eigenvalue $+1$ if $[E,t]=0$ and eigenvalue $-1$ if $\{E,t\}=0$, up to the irrelevant global phases of Pauli operators. Therefore, each bilocal Pauli observable tests a commutation sign of the error and the check. The success event is the condition that $E$ commutes with every generator $g_j$ of $S_U$, corresponding to the trivial-syndrome projector $\prod_{j=2}^n(I+g_{j,A}\otimes g_{j,B}^*)/2=|S_U|^{-1}\sum_{s\in S_U}s_A\otimes s_B^*$. Then the success probability given a general $n$-pair joint input state $\rho$ is $p_{\mathrm{succ}} = \Pr(E\in N(S_U)) = |S_U|^{-1}\sum_{s\in S_U}\Tr[(s_A\otimes s_B^*)\rho]$. The numerator of the successful output fidelity is $p_{\mathrm{succ}}F_{\mathrm{succ}} = \Pr(E\in S_U)$. This is represented by averaging over the full normalizer, because for a stabilizer code $S_U$ is the set of Pauli operators that commute with every element of $N(S_U)$. If $E\in S_U$ then $[E,h]=0$ for all $h\in N(S_U)$, so every $h_A\otimes h_B^*$ correlation is $+1$ and so is the average over all $h\in N(S_U)$. If $E\notin S_U$, the $h_A\otimes h_B^*$ correlation $\pm 1$ will exist in pairs and cancel out in the average. Therefore, the numerator of successful output fidelity for a general joint input state $\rho$ is evaluated as $|N(S_U)|^{-1}\sum_{h\in N(S_U)}\Tr[(h_A\otimes h_B^*)\rho]$.

We now define the single-pair correlations that evaluate these averages. For a single BDS input, define $\lambda_i(P) = \Tr\!\left[(P_A\otimes P_B^*)\rho_i\right]$ for $P\in\{I,X,Y,Z\}$, where the subscripts denote the party that the Pauli operator applies on. Then $\lambda_i(I)=1$, $\lambda_i(X)=p_i(I)+p_i(X)-p_i(Y)-p_i(Z)$, $\lambda_i(Y)=p_i(I)-p_i(X)+p_i(Y)-p_i(Z)$, and $\lambda_i(Z)=p_i(I)-p_i(X)-p_i(Y)+p_i(Z)$. The signs in $\lambda_i(X),\lambda_i(Y),\lambda_i(Z)$ are obtained by evaluating the bilocal Pauli observable $P_A\otimes P_B^*$ on the four Bell states produced by Pauli errors on one half of $\Phi^+$. For example, $X_A\otimes X_B^*$ has eigenvalue $+1$ on the $I$- and $X$-error Bell states and eigenvalue $-1$ on the $Y$- and $Z$-error Bell states, giving the sign pattern in $\lambda_i(X)$. This is also known as the the Walsh-Hadamard transform~\cite{flammia2020efficient} which connects the Pauli error rates $p_i(P)$ and the Pauli eigenvalues $\lambda_i(P)$. For an $n$-qubit Pauli operator $t=t_1\otimes\cdots\otimes t_n$, independence gives $\Tr[(t_A\otimes t_B^*)\bigotimes_i\rho_i]=\prod_i\lambda_i(t_i)$.

Taking expectation values of the stabilizer and normalizer averages, and using the product form of the BDS inputs, gives
\begin{align}
    p_{\mathrm{succ}} &= \frac{1}{|S_U|} \sum_{s\in S_U} \prod_{i=1}^n \lambda_i(s_i), \label{eq:bicep-main-succ}\\
    p_{\mathrm{succ}}F_{\mathrm{succ}} &= \frac{1}{|N(S_U)|} \sum_{h\in N(S_U)} \prod_{i=1}^n \lambda_i(h_i). \label{eq:bicep-main-weighted}
\end{align}
We call $p_{\mathrm{succ}}F_{\mathrm{succ}}$ the success-weighted output Bell fidelity, and may use $\tilde{F}_{\mathrm{succ}}$ to denote it.

For an entanglement source that distributes Werner state $\rho^{(W)}(F_i)$, set $q_i=(1-F_i)/3$ and we have $p_i(I)=F_i$, $p_i(X)=p_i(Y)=p_i(Z)=q_i$. Therefore, $\lambda_i(X) = \lambda_i(Y) = \lambda_i(Z) = F_i-q_i = w_i$, where $w_i$ is the visibility (Werner parameter) of source $i$ as previously defined, while $\lambda_i(I)=1$. These Pauli eigenvalues correspond to a 1-qubit depolarizing channel, as a Werner state is effectively the result of applying a 1-qubit depolarizing channel to one party of $|\Phi^+\rangle$. We also define $w=(w_1,\dots,w_n)$ as the visibility vector for the $n$ Werner sources. It is safe to take $0\leq w_i\leq 1$, which is equivalent to $F_i\in[1/4,1]$, because in practice we deal with entangled Werner states with $F_i\geq 1/2$, which makes non-negative visibility automatic. Since each factor in Eqs.~\eqref{eq:bicep-main-succ} and~\eqref{eq:bicep-main-weighted} is either $1$ or one of the $w_i$, both figures of merit are multi-affine polynomials in $w$ with non-negative coefficients.

\textit{General biCEP theorem}---
Then we have the general theorem for how AS enhance the performance of arbitrary $n$-to-1 biCEP with arbitrary Werner states as input.
\begin{theorem}[AS enhancement for biCEPs with Werner sources]\label{thm:bicep-as-main}
    Fix any $n$-to-1 biCEP and any $n$ Werner sources with visibilities $0\leq w_i\leq 1$. Compare the non-AS random-permutation baseline of Eq.~\eqref{eq:noas-baseline} with AS after any finite number $m\geq 1$ of distribution rounds. Then the expected per-package success probability and success-weighted output Bell fidelity satisfy
    \begin{align}
        p_{\mathrm{succ,AS}}^{(m)} &\geq p_{\mathrm{succ,noAS}},\\
        \tilde{F}_{\mathrm{succ,AS}}^{(m)} &\geq \tilde{F}_{\mathrm{succ,noAS}}.
    \end{align}
    Moreover, we have monotonicity with the number of accumulation rounds $p_{\mathrm{succ,AS}}^{(m+1)} \ge p_{\mathrm{succ,AS}}^{(m)}$ and $\tilde{F}_{\mathrm{succ,AS}}^{(m+1)} \ge \tilde{F}_{\mathrm{succ,AS}}^{(m)}$.
\end{theorem}

The AS enhancement also holds in the asymptotic limit $m\to\infty$, which has a clean majorization picture. For $x,y\in\mathbb R^n$, write $x^\downarrow$ for the decreasing rearrangement of $x$, where $x^\downarrow_i$ is the $i$-th largest entry of $x$. We say that $x$ is majorized by $y$, written $x\prec y$, when $\sum_{i=1}^r x_i^\downarrow\leq\sum_{i=1}^r y_i^\downarrow$ for $r=1,\dots,n-1$ and $\sum_i x_i=\sum_i y_i$. In the limit $m\to\infty$, Theorem~\ref{thm:rho_eff} means that the Werner visibility vector $w$ becomes $w_\infty = (\bar{w},\dots,\bar{w})$ with $\bar{w}=\frac{1}{n}\sum_i w_i$, and $w_\infty\prec w$ for every valid visibility vector $w$. Given the averaging over the random source-label permutation in the non-AS baseline, Eqs.~\eqref{eq:bicep-main-succ} and~\eqref{eq:bicep-main-weighted} become \textit{symmetric} multi-affine polynomials with non-negative coefficients, hence non-negative linear combinations of elementary symmetric polynomials. Here the elementary symmetric polynomial of degree $r$ is $e_r(w_1,\dots,w_n) = \sum_{\substack{J\subseteq \{1,\dots,n\}, |J|=r}} \prod_{j\in J} w_j$ while for degree zero $e_0=1$. For instance, $e_2(w_1,w_2,w_3)=w_1w_2+w_1w_3+w_2w_3$. As we explicitly prove in the Supplemental Material (Sec.~\ref{sec:biCEP_AS})~\cite{supmat}, non-negative linear combinations of elementary symmetric polynomials are Schur-concave on the non-negative orthant $w\in \R_+^n$~\cite{marshall1979inequalities}, so these symmetrized figures of merit increase when $w$ is replaced by the uniform vector $w_\infty$.

For finite $m$, we can expand the symmetrized biCEP figure of merit expressions into visibility monomials and compare each monomial degree separately. A degree-$r$ visibility monomial is a product of $r$ visibilities, arising from a Pauli term whose non-identity support has size $r$. In the non-AS random-permutation baseline, the $r$ source labels are uniform randomly chosen from the $n$ sources, so the average contribution of such a monomial is the elementary symmetric mean $e_r(w)/\binom{n}{r}$. Under AS, the physical pairs in the shuffled package are sampled without replacement from the $mn$ accumulated pairs corresponding to the accumulated visibility vector $w^{(m)} = (\underbrace{w_1,\dots,w_1}_{m}, \dots, \underbrace{w_n,\dots,w_n}_{m})$. Hence the AS average is $e_r(w^{(m)})/\binom{mn}{r}$. Maclaurin's inequality and Vandermonde's identity imply $e_r(w^{(m)})/\binom{mn}{r} \ge e_r(w)/\binom{n}{r}$ for $r=0,\dots,n$, so every degree-$r$ contribution is no smaller after AS. Since the symmetrized biCEP expressions of success probability and success-weighted output fidelity have only non-negative coefficients, adding these inequalities proves the finite-$m$ comparison. See the Supplemental Material (Sec.~\ref{sec:biCEP_AS})~\cite{supmat} for the full proof, and also for the proof of monotonicity.

\textit{Conclusion and discussion}---
We have shown that shared randomness can enhance entanglement purification through a symmetrization mechanism. Accumulating and shuffling (AS) requires only buffer memories and classical shared randomness, yet for \textit{arbitrary} $n$ Werner sources and \textit{arbitrary} fixed $n$-to-1 biCEPs it guarantees enhancement of success probability and success-weighted Bell fidelity over the non-AS baseline. The AS strategy is exceptionally simple and inexpensive, as it does not require characterization of received entangled states or modification of the EPP circuit itself. The results also establish a novel mechanism and instance of how classical randomness can enhance the performance of a quantum task and offer insightful implications. 

While we focus on Werner sources in the main text to provide clean proofs of enhancement induced by the AS strategy, the AS enhancement works beyond Werner sources. In the End Matter, we demonstrate the tolerance to error biases in the input entangled states using the bilocal CNOT EPP as example, which also highlights the importance of input states' structure. In practice, we will inevitably face memory decoherence whenever we need to buffer quantum states. Also in the End Matter, we demonstrate the robustness to memory decoherence of the AS strategy, considering memory depolarization, Werner sources and the bilocal CNOT EPP as example. Moreover, even though we primarily assume that the entanglement sources will each distribute a fixed density operator, we prove in the Supplemental Material (Sec.~\ref{sec:source_ensemble})~\cite{supmat} that our analysis framework using the effective state of each EPP input package can apply to the more general scenario where each source produces a density matrix drawn from some probability distribution in one entanglement distribution round.

Theorem~\ref{thm:bicep-as-main} addresses success-weighted fidelity rather than normalized successful output fidelity intentionally, because normalized successful output fidelity is \textit{not universally} improved for arbitrary biCEPs with AS, even though this is the case for the bilocal CNOT EPP as shown in Proposition~\ref{thm:succ_improvement}. For example, as we remark in the Supplemental Material (Sec.~\ref{sec:biCEP_AS})~\cite{supmat}, there exist combinations of 5 Werner sources such that AS will lower $F_{\mathrm{succ}}$ for the $[[5,1,3]]$-code EPP. 

Furthermore, this work opens up broad opportunities for future research. For instance, it will be exciting to elucidate the role of shared randomness and quantum memories in the discovered performance enhancement from resource theory perspectives~\cite{forster2009distilling,buscemi2012all,rosset2018resource,chitambar2019quantum,schmid2020type}. It is also interesting to explore potential connections between the AS strategy and activation~\cite{brandao2008correlated}, renormalization~\cite{waeldchen2016renormalizing}, and scrambling~\cite{gu2026constant}, and whether the AS strategy can be derandomized.

\textit{Acknowledgments}---
We thank Senrui Chen, Henry Minzhao Liu, Shifan Xu, and Han Zheng for inspiring discussions. Part of this research was performed while the authors were visiting the Institute for Mathematical and Statistical Innovation (IMSI), which is supported by the National Science Foundation (Grant No. DMS-1929348).
We acknowledge support from the ARO MURI (W911NF-21-1-0325), NSF (ERC-1941583). This material is based upon work supported by the U.S. Department of Energy, Office of Science, National Quantum Information Science Research Centers and Advanced Scientific Computing Research (ASCR) program under contract number DE-AC02-06CH11357 as part of the InterQnet quantum networking project. We also acknowledge support from the NSF Quantum Leap Challenge Institute for Hybrid Quantum Architectures and Networks (NSF Grant No. 2016136), and the Marshall and Arlene Bennett Family Research Program.

\bibliography{references}

\section*{End Matter}

\begin{figure*}[t]
    \centering
    \includegraphics[width=0.95\linewidth]{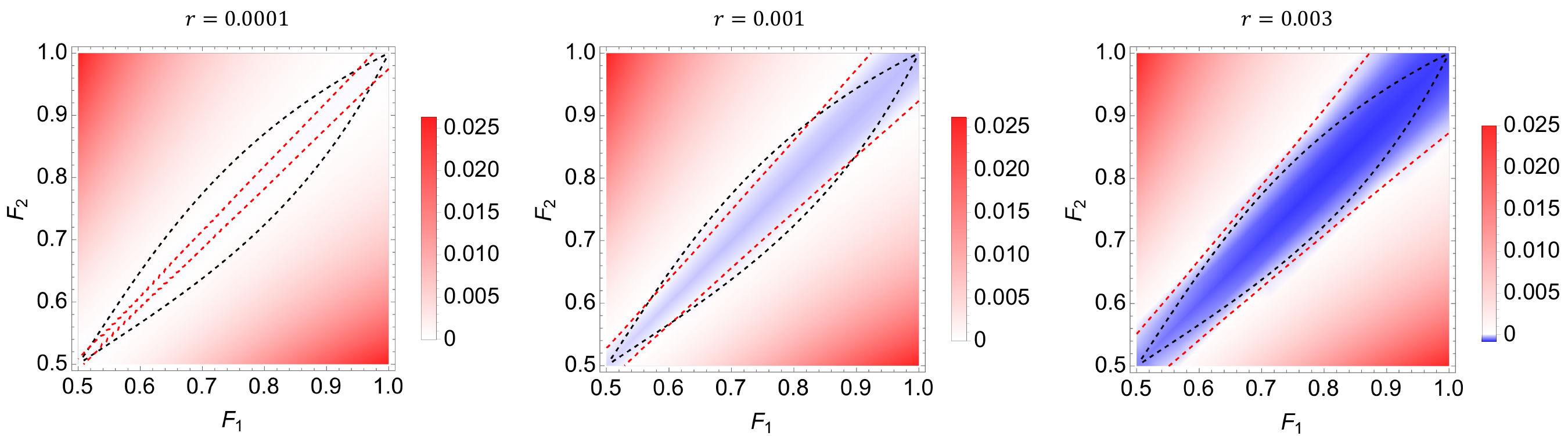}
    \caption{Difference between $F_{\mathrm{succ}}(\rho_{\mathrm{eff},n=2}^{(2)})$ and $F_{\mathrm{succ}}\left(F_1,F_2\right)$ for $r=0.0001,0.001,0.003$. White color is fixed to represent zero value, while blue color denotes negative value and red color corresponds to positive value. The red dashed lines illustrate the boundaries where the difference is exactly zero. The black dashed lines illustrate the boundaries where $F_{\mathrm{succ}}\left(F_1,F_2\right)-\max\{F_1,F_2\}$ is exactly zero.}
    \label{fig:m=2vsNoAS_fid}
\end{figure*}

\textit{Considerations and justifications of the operational assumptions}---
Our operational assumptions of the considered entanglement distribution scenario are mainly two-fold: generally non-identical (heterogeneous) entanglement sources, and the unavailability of the source label of each entangled state to both parties that receive the entangled states. 

The first assumption of heterogeneity is natural, since in practice even the nominally identical devices and processes that establish remote entanglement will still almost always be subject to differences. Moreover, it is also possible that the entanglement sources or entanglement generation processes may undergo temporal drifts~\cite{proctor2020detecting}. This may not only make sources different from each other at a specific time point, but can also make a specific source different from its past self. The potential temporal drift further motivates and justifies the consideration of input-independent strategies to enhance EPP performance, such as our AS strategy. This is because even though there may exist efficient methods to extract useful information of the entangled states~\cite{haah2016sample,huang2020predicting}, the resource cost (measured entangled states) of such methods is not one-time but recurrent.

For the assumption of not knowing which source each entangled state is from, we may look into quantum network architectures~\cite{kozlowski2023rfc}. Even though quantum network architecture has not yet come to a global consensus, it is already a well-established idea to separate the physical implementation layer from higher layers via abstraction. The information of which source each entangled state comes from will be classical metadata that needs additional maintenance~\cite{kozlowski2023rfc}. In general, if the entangled states distributed from the entanglement sources have qualities (not necessarily identical) that are above a certain minimum requirement, they can be simply accepted and utilized for subsequent manipulation or applications, without being distinguishable to the receivers~\cite{dahlberg2019link,pompili2022experimental,kozlowski2023rfc}. In addition, the which-source information may serve as side-channel information that can potentially affect security analyses for some security and cryptography application scenarios~\cite{huang2018quantum}.

\textit{Arbitrary Bell-diagonal average input state for bilocal CNOT EPP}---
We have assumed ignorance of the source labels for each state in each entanglement distribution round. Therefore, even if the two parties can perform state characterization they can at most know the average input state (AIS) $\bar{\rho}$. Generally, $\rho_1$, $\rho_2$ and $\bar{\rho}$ are not necessarily in Werner form, and the knowledge of the principal error component of $\bar{\rho}$ could potentially help us further improve the EPP performance: We may use local operations and classical communication (LOCC)~\cite{chitambar2014everything} to permute the error components so that the principal errors component of the AIS is detectable by the EPP. We thus consider an AIS $\bar{\rho}$ characterized by four density matrix diagonal elements in the Bell basis $(\bar{\nu}_1,\bar{\nu}_2,\bar{\nu}_3,\bar{\nu}_4)$, s.t. $\bar{\nu}_j\geq 0$ and $\sum_{j=1}^4\bar{\nu}_j=1$ while $\bar{\nu}_2 \leq \bar{\nu}_3,\bar{\nu}_4 \leq \bar{\nu}_1$, and $\bar{\nu}_1>1/2$ so that the AIS is entangled.

As the DEJMPS protocol cannot detect errors corresponding to $\Phi^-$, we first require the allowed input states to have identical bias of the undetectable error $a$ as the AIS. Explicitly, we consider $\bar{\rho}=(\bar{F},a(1-\bar{F}),b(1-\bar{F}),(1-a-b)(1-\bar{F})) = (\rho_1+\rho_2)/2$ with $a\in[0,1/3]$, while $\rho_1=(F_1,a(1-F_1),b_1(1-F_1),(1-a-b_1)(1-F_1))$. Then we have the following on AS enhancement, whose proof can be found in the Supplemental Material (Sec.~\ref{sec:bilocal_cnot_epp_results})~\cite{supmat}.
\begin{proposition}\label{thm:ais_bias}
    Let $\bar{\rho}=(\bar{F},a(1-\bar{F}),b(1-\bar{F}),(1-a-b)(1-\bar{F}))$ with $a\in[0,1/3]$, and let $\rho_1=(F_1,a(1-F_1),b_1(1-F_1),(1-a-b_1)(1-F_1))$ with $\rho_2=2\bar{\rho}-\rho_1$. For the DEJMPS protocol we have
    \begin{align}
        p_{\mathrm{succ}}(\rho_{\mathrm{eff},n=2}^{(m)}) &\geq p_{\mathrm{succ}}(\rho_1\otimes\rho_2),\\
        E_{\mathrm{succ}}(\rho_{\mathrm{eff},n=2}^{(m)}) &\geq E_{\mathrm{succ}}(\rho_1\otimes\rho_2)
    \end{align}
    for all $\bar{F}\in[1/2,1]$ if $a\in[2-\sqrt{3},1/3]$, where $E_{\mathrm{succ}}$ is any entanglement quality metric that is monotone increasing with Bell fidelity for the successful output BDS.
\end{proposition}

Many common entanglement quality metrics such as concurrence~\cite{hill1997entanglement,wootters1998entanglement}, (logarithmic) negativity~\cite{plenio2005logarithmic}, and a common upper bound of distillable entanglement, i.e. the Rains' bound~\cite{rains1999bound,rains2001semidefinite} or relative entropy of entanglement~\cite{vedral1997quantifying,vedral1998entanglement,vedral2002role}, increase monotonically with Bell fidelity for BDS.

These results demonstrate the guaranteed advantage of the AS strategy given input states that have identical, low error bias. On the other hand, when the minor error component is small, i.e. when the error is highly biased, AS may be unnecessary since the EPP already excels in purifying BDS with biased error. In addition, if $\rho_1$ and $\rho_2$ are known, deterministic state-dependent strategies of entanglement purification may also perform better. As shown below, even if both BDS inputs do not share the same $a$, the success probability and success-weighted output fidelity are still guaranteed to improve by the AS strategy over the non-AS baseline. Nevertheless, if they do not share the same $a$ there will be cases when the (normalized) successful output fidelity is not improved. 

We also consider the most general case where $\rho_1$ and $\rho_2$ can be arbitrary entangled BDS, i.e. $F(\rho_1),F(\rho_2)>1/2$, as long as they satisfy $(\rho_1+\rho_2)/2=\bar{\rho}$. Here the success probability and success-weighted output fidelity, are guaranteed to be improved by accumulating and shuffling, regardless of the specific forms of $\rho_1$ and $\rho_2$.
\begin{proposition}\label{thm:arb_bds_succ_prob_norm_fid}
    For an arbitrary two-source BDS average input state $\bar{\rho}=(\rho_1+\rho_2)/2$, the two-copy DEJMPS protocol satisfies
    \begin{align}
        p_{\mathrm{succ}}(\rho_{\mathrm{eff},n=2}^{(m)}) &\geq p_{\mathrm{succ}}(\rho_1\otimes\rho_2),\\
        \tilde{F}_{\mathrm{succ}}(\rho_{\mathrm{eff},n=2}^{(m)}) &\geq \tilde{F}_{\mathrm{succ}}(\rho_1\otimes\rho_2).
    \end{align}
\end{proposition}
The proof can be found in the Supplemental Material (Sec.~\ref{sec:bilocal_cnot_epp_results})~\cite{supmat}. On the other hand, the sign of the difference between $F_{\mathrm{succ}}(\bar{\rho}\otimes\bar{\rho})$ and $F_{\mathrm{succ}}(\rho_1\otimes\rho_2)$ depends on input state parameters, which demonstrates the importance of the minor error component of the input BDS to the DEJMPS protocol.

\textit{Memory decoherence effect on AS for bilocal CNOT EPP}---
In practice, entanglement distribution may take significantly longer time in comparison to local operations, and quantum memories that store entangled states inevitably decohere when storing the entangled states waiting for later entangled states to arrive. Therefore, we now include the quantum memory idling decoherence. Without loss of generality, we assume that initial noisy entangled states are Werner states, and all quantum memories undergo symmetric depolarizing channel with identical depolarizing rate $\kappa$ s.t. the completely depolarizing probability is $1 - e^{-\kappa t}$. Then the Bell state fidelity evolves as $F(F_0,t) = F_0 e^{-2\kappa t} + (1-e^{-2\kappa t})/4$~\cite{zang2025entanglement}, where $F_0$ is the initial fidelity of the Werner state when it is first stored in quantum memories, and $t$ is the time for both quantum memories to idle.

Consider that two parties accumulate over $m$ rounds of entanglement distribution and store $2m$ entangled states in total. The effective input state is the average of all possible tensor products of two non-identical entangled states in the ensemble of $2m$ accumulated states, i.e. $\rho_{\mathrm{eff},n=2}^{(m)} = [\sum_{i=1,2}\sum_{0\leq j<k<m}\rho_i(j)\otimes\rho_i(k) + \sum_{j,k=0}^{m-1}\rho_1(j)\otimes\rho_2(k)]/\binom{2m}{2}$, where $\rho_i(s), i=1,2$ denotes the final state after the initial entangled state $\rho_i(0)$ undergoes $sT$ duration of memory decoherence, recalling that $T$ is the cycle time of entanglement distribution. The above expression follows from the identity $m\mathcal{N}_\mathrm{pairing}(2m) = \binom{2m}{2}\mathcal{N}_\mathrm{pairing}(2m-2)$, where $\mathcal{N}_\mathrm{pairing}(2m) = \frac{(2m)!}{2^m m!}$ is the number of different ways to create $m$ pairs of states out of the $2m$ states, without over-counting different orders in the random pairing process. The interpretation is that there are $\mathcal{N}_\mathrm{pairing}(2m-2)$ ways of pairing up the $2m$ accumulated states which include a certain pair of non-identical entangled states, and in total there are $\binom{2m}{2}$ possible pairs of non-identical entangled states. If we expand all these possible pairings, the ensemble of DEJMPS input entangled state pairs will contain $\mathcal{N}_\mathrm{pairing}(2m-2)$ copies of all $\binom{2m}{2}$ possible pairs of non-identical entangled states.

Intuitively, when the decoherence rate is very small, the AS performance will be very similar to the noiseless case according to continuity. However, when the $\rho_1$ and $\rho_2$ are very similar to each other or even identical, as long as there is non-zero decoherence, entanglement purification with AS will very likely have worse performance than without AS. In addition, when the decoherence is not negligible, it should suffice to use $m=2$ AS, because for larger $m$ the accumulated entangled states will decohere more significantly. We also expect that AS will completely lose any advantage after the decoherence rate increases above a certain threshold that depends on error model and the figure of merit under examination. See more details in the Supplemental Material (Sec.~\ref{sec:decoherence_model})~\cite{supmat}. Then without loss of generality, we compare $m=2$ AS given Werner input and memory depolarization, to the case without AS. Specifically, we can prove the following.
\begin{proposition}\label{thm:m=2<NoAS_with_decoherence}
    For raw Werner states with arbitrary fidelities $F_1,F_2\in[1/2,1]$ and memory depolarizing channels characterized by $r=2\kappa T$, the two-round AS package satisfies
    \begin{align}
        p_{\mathrm{succ}}(\rho_{\mathrm{eff},n=2}^{(2)}) &\leq p_{\mathrm{succ}}(F_1,F_2) \quad \text{for } r\geq 0.206,\\
        F_{\mathrm{succ}}(\rho_{\mathrm{eff},n=2}^{(2)}) &\leq F_{\mathrm{succ}}(F_1,F_2) \quad \text{for } r\geq 0.106,\\
        \tilde{F}_{\mathrm{succ}}(\rho_{\mathrm{eff},n=2}^{(2)}) &\leq \tilde{F}_{\mathrm{succ}}(F_1,F_2) \quad \text{for } r\geq 0.150.
    \end{align}
\end{proposition}
The proof can be found in Supplemental Material (Sec.~\ref{sec:decoherence_model})~\cite{supmat}. We further visualize the comparison between $F_{\mathrm{succ}}(\rho_{\mathrm{eff},n=2}^{(2)})$ and $F_{\mathrm{succ}}\left(F_1,F_2\right)$ for different values of $r$ in Fig.~\ref{fig:m=2vsNoAS_fid}. We observe as expected that for near-identical raw states $\rho_1$ and $\rho_2$, AS cannot offer better performance even under low memory decoherence. The region where there is no advantage expands as the memory decoherence rate increases, while the advantage when two input states are different remains robust for relatively large decoherence rate up to $r\approx 0.1$. It is noteworthy that when using AS under low memory decoherence rate, it is possible to achieve higher successful output fidelity when the fidelity without AS is already above the higher input fidelity, as shown by the region outside the red dashed boundary while within the black dashed boundary, demonstrating further improvement over the better raw state.

\clearpage
\onecolumngrid


\tableofcontents

\section{Effective single-package state after accumulating and shuffling}\label{sec:eff_state}
Recall the setup of the accumulating and shuffling scenario. Consider $n$ states $\rho_1,\rho_2,\dots,\rho_n$. We prepare $m$ copies of each state, so we have $mn$ states in total. Before shuffling, one possible ordering of the full state is
\begin{align}
    \rho_{\mathrm{all}} = \rho_1^{\otimes m} \otimes \rho_2^{\otimes m} \otimes \dots \otimes \rho_n^{\otimes m}.
\end{align}
We apply a uniformly random permutation on the $mn$ states and then divide the shuffled states into $m$ packages of size $n$. By symmetry of the random shuffle, every package has the same marginal distribution, i.e. the expected state of the package. It is thus sufficient to study the first package without loss of generality.

Let $A=(A_1,\dots,A_n)$ be the random label vector of the first package, where $A_j=i$ means that the $j$-th slot of the package contains a copy of $\rho_i$. Conditioned on $A=a=(a_1,\dots,a_n)$, the state of the package is $\rho_a = \rho_{a_1}\otimes\dots\otimes\rho_{a_n}$. The expected, or effective, single-package state is
\begin{align}
    \rho_{\mathrm{eff}}^{(m)} = \E_A[\rho_A].
\end{align}

\subsection{Ordered sampling perspective}

For a label vector $a=(a_1,\dots,a_n)\in \{1,\dots,n\}^n$, denote
the number of times when label $i$ occurs in $a$ as $c_i(a)$.

\begin{lemma}
    For any ordered label vector $a=(a_1,\dots,a_n)$,
    \begin{align}
        \Pr(A=a) = \frac{\prod_{i=1}^n (m)_{c_i(a)}}{(mn)_n},
    \end{align}
    where we have used the notation of falling factorial $(x)_r=x(x-1)\dots(x-r+1)$ with $(x)_0=1$.
\end{lemma}

\begin{proof}
Without loss of generality, we focus on the first package obtained by drawing $n$ systems without replacement from a list containing $m$ copies of each label $1,\dots,n$. Suppose the first $t-1<n$ labels have already been drawn. If label $i$ has appeared $r_i$ times so far, then there are $m-r_i$ remaining copies of label $i$, and there are $mn-(t-1)$ systems remaining in total. Thus the probability of drawing the label $a_t$ at step $t$, given the previous labels, is
\begin{align}
    \Pr(A_t=a_t | A_1=a_1,\dots,A_{t-1}=a_{t-1}) = \frac{m-\#\{s<t: a_s=a_t\}}{mn-(t-1)}.
\end{align}
Multiplying these conditional probabilities over $t=1,\dots,n$, we have
\begin{align}
    \Pr(A=a) &= \prod_{t=1}^n \Pr(A_t=a_t | A_1=a_1,\dots,A_{t-1}=a_{t-1}) \notag\\
    &= \prod_{t=1}^n \frac{m-\#\{s<t: a_s=a_t\}}{mn-(t-1)} \notag\\
    &= \frac{\prod_{t=1}^n\left( m-\#\{s<t: a_s=a_t\} \right)}{(mn)_n}.
\end{align}

We now simplify the numerator. For each label $i\in\{1,\ldots,n\}$, define $T_i=\{t\in\{1,\ldots,n\}:a_t=i\}$. The sets $T_1,\ldots,T_n$ form a disjoint partition of $\{1,\ldots,n\}$. Hence
\begin{align}
    \prod_{t=1}^n \left( m-\#\{s<t:a_s=a_t\} \right) &= \prod_{i=1}^n \prod_{t\in T_i} \left( m-\#\{s<t:a_s=a_t\} \right) \notag\\
    &= \prod_{i=1}^n \prod_{t\in T_i} \left( m-\#\{s<t:a_s=i\} \right),
\end{align}
where in the second line we used the fact that $a_t=i$ for $t\in T_i$. Now write the elements of $T_i$ in increasing order $T_i=\{t_{i,1}<t_{i,2}<\cdots<t_{i,c_i(a)}\}$. At the $q$-th occurrence $t_{i,q}$, the label $i$ has appeared exactly $q-1$ times. Therefore, $\#\{s<t_{i,q}:a_s=i\}=q-1$. It follows that
\begin{align}
    \prod_{t\in T_i} \left( m-\#\{s<t:a_s=i\} \right) = \prod_{q=1}^{c_i(a)} \left(m-(q-1)\right) = (m)_{c_i(a)}.
\end{align}
Combining the contributions from all labels gives
\begin{align}
    \prod_{t=1}^n \left( m-\#\{s<t:a_s=a_t\} \right) = \prod_{i=1}^n (m)_{c_i(a)}.
\end{align}
With the simplified numerator and denominator, we have obtained the desired expression of this lemma.
\end{proof}

Now we have the explicit formula for the expected state of an $n$-state package after the random shuffle, written with explicit ordering of the $n$ states.
\begin{theorem}[Effective state, ordered form]
The expected state of any package is
\begin{align}
    \rho_{\mathrm{eff}}^{(m)} = \sum_{a_1,\dots,a_n=1}^n \frac{\prod_{i=1}^n (m)_{c_i(a)}}{(mn)_n} \rho_{a_1}\otimes\dots\otimes\rho_{a_n}.
\end{align}
Here $a=(a_1,\dots,a_n)$. Terms with $c_i(a)>m$ automatically have
zero weight because $(m)_{c_i(a)>m}=0$.
\end{theorem}

\begin{proof}
By definition the expected state is
\begin{align}
    \rho_{\mathrm{eff}}^{(m)} = \E_A[\rho_A] = \sum_{a_1,\dots,a_n=1}^n \Pr(A=a) \rho_{a_1}\otimes\dots\otimes\rho_{a_n}.
\end{align}
Then using the above Lemma we immediately get the formula.
\end{proof}

\subsection{Occupation number perspective}

The previous derivation is based on the most straightforward operational interpretation of the random shuffle process. A more compact expression can be obtained by grouping terms according to occupation numbers, i.e. the number of copies of $\rho_i$ in the package. We define the occupation number vector $k=(k_1,\dots,k_n)$, where $k_i$ is the number of copies of $\rho_i$ in the package, with $k_i\geq 0$ and $\sum_{i=1}^n k_i=n$. Since only $m$ copies of each state exist, we also require $k_i\leq m$. Now fix such an occupation vector $k$, and define the reference ordered tensor product
\begin{align}
  X_k = \rho_1^{\otimes k_1} \otimes \rho_2^{\otimes k_2} \otimes \dots \otimes \rho_n^{\otimes k_n} .
\end{align}
If $k_i=0$, the corresponding factor in the tensor product is simply omitted. Let $S_n$ be the symmetric group of the $n$ slots in a package. For permutation $\pi\in S_n$, let $U_\pi$ be the unitary that permutes tensor factors according to $\pi$. Then we can define the symmetrized (permutation-twirled) state
\begin{align}
    \Sym(X) = \frac{1}{n!}\sum_{\pi\in S_n} U_\pi X U_\pi^\dagger.
\end{align}
After we know only the occupation numbers $k$, the random shuffle still randomizes the positions of those states inside the package. Conditioned on $K=k$, all distinct orderings of the list with $k_i$ copies of label $i$ are equally likely. The average over those orderings gives exactly $\Sym(X_k)$. Then we can also express the expected state of an $n$-state package after the random shuffle with symmetrized states (thus without explicit ordering).

\begin{theorem}
The expected state of any package can also be written as
\begin{align}
    \rho_{\mathrm{eff}}^{(m)} = \sum_{\substack{k_1+\dots+k_n=n\\0\leq k_i\leq m}} \frac{\prod_{i=1}^n \binom{m}{k_i}}{\binom{mn}{n}} \Sym \left( \rho_1^{\otimes k_1} \otimes\dots\otimes \rho_n^{\otimes k_n} \right).
\end{align}
\end{theorem}

\begin{proof}
We can group the ordered vectors $a=(a_1,\dots,a_n)$ by their occupation vector $K(a)=(c_1(a),\dots,c_n(a))$. For a fixed occupation vector $k$, every ordered vector with that occupation vector has the same probability, since the probability for the ordered vector $\prod_i (m)_{k_i}/(mn)_n$ depends only on $k$. Consequently, for random package state $\rho_A$ with random ordered label vector $A=(a_1,\dots,a_n)$, we have $\E[\rho_A\mid K=k]=\Sym(X_k)$.

The probability of the occupation vector itself is
\begin{align}
    \Pr(K=k) = \frac{\prod_i \binom{m}{k_i}}{\binom{mn}{n}},
\end{align}
because one chooses $k_i$ of the $m$ available copies of type $i$, for each $i$, out of all $\binom{mn}{n}$ possible unordered selections of $n$ systems. Then the desired formula is a direct result of the law of total expectation
\begin{align}
    \rho_{\mathrm{eff}}^{(m)} &= \sum_k \Pr(K=k) \E[\rho_A\mid K=k] \notag\\
    &= \sum_k \frac{\prod_{i=1}^n \binom{m}{k_i}}{\binom{mn}{n}} \Sym \left( \rho_1^{\otimes k_1} \otimes\dots\otimes \rho_n^{\otimes k_n} \right),
\end{align}
where valid $k$ satisfies $k_1+\dots+k_n=n$ and $0\leq k_i\leq m$.
\end{proof}

\subsection{Permutation twirling perspective}
The same expected state can be described using the full permutation group on all $mn$ systems. Let $S_{mn}$ be the group of permutations of the full list (the symmetric group), and let $V_\pi$ be the unitary representation that permutes the $mn$ tensor factors. The globally shuffled state is the group average (result of permutation twirling)
\begin{align}
    \mathcal T_{S_{mn}}(\rho_{\mathrm{all}}) = \frac{1}{(mn)!} \sum_{\pi\in S_{mn}} V_\pi \rho_{\mathrm{all}} V_\pi^\dagger.
\end{align}
Let $\Tr_{\mathrm{rest}}$ denote the partial trace over all systems except one chosen package, then
\begin{align}
    \rho_{\mathrm{eff}}^{(m)} = \Tr_{\mathrm{rest}} \left[ \mathcal T_{S_{mn}}(\rho_{\mathrm{all}}) \right].
\end{align}
Note that according to the permutation invariance after the random shuffle, the chosen package can be arbitrary.

\section{Bilocal Clifford entanglement purification (biCEP)}\label{sec:bicep}
In this section, we explain the stabilizer-code interpretation of biCEP protocols and derive figures of merit in more details, and demonstrate the stabilizer-code analysis of example biCEP protocols.

\subsection{Stabilizer code interpretation of biCEP}
Let $S\subseteq \Pauli_n$ be an abelian (Pauli) stabilizer subgroup not containing $-I^{\otimes n}$. The common $+1$ eigenspace of $S$ is a stabilizer codespace. An $[[n,1,d]]$ stabilizer code has $\lvert S\rvert=2^{n-1}$ and normalizer (also the centralizer, as we work within the Pauli group) $N(S)\subseteq\Pauli_n$ of size $\lvert N(S)\rvert=4\lvert S\rvert$, with $N(S)/S \cong \Pauli_1$ acting as the logical Pauli group.

Let $S_0 \coloneqq \langle Z_2,\dots,Z_n\rangle$ be the trivial $[[n,1]]$ stabilizer code whose logical qubit before encoding and after decoding is the first qubit. Note that we will use $S_n,n\ge 1$ to denote symmetric groups, and here $S_0$ is not to be confused with such symmetric groups. The subscripts of the Pauli operators index the qubits in the stabilizer code $i=1,2,\dots,n$, where the first qubit corresponds to the logical qubit. Given a Clifford $U$ for biCEP, define
\begin{align}\label{eq:SU_def}
    g_j \coloneqq U^\dagger Z_j U,~~S_U \coloneqq \langle g_2,\dots,g_n\rangle = U^\dagger S_0 U .
\end{align}
Since $U$ is Clifford, $S_U$ is still an $[[n,1]]$ stabilizer group.

Comparing Alice's and Bob's computational-basis outcomes on sacrificed pair $j$ after applying $U\otimes U^*$ is equivalent to measuring the bilocal Pauli observable $Z_{A,j}\otimes Z_{B,j}$ after the Clifford. Propagating this observable back through $U\otimes U^*$ gives
\begin{equation}\label{eq:pulled_back_bilocal_stabilizer}
    (U^\dagger Z_j U)_A \otimes ((U^*)^\dagger Z_j U^*)_B = g_{j,A}\otimes g_{j,B}^* .
\end{equation}
Thus the computational-basis parity checks in the standard biCEP circuit are exactly the bilocal stabilizer measurements associated with $S_U$, with the usual complex conjugation sign convention on Bob's side. Then we have the following Lemma which formalizes the interpretation of biCEP in terms of stabilizer codes.

\begin{lemma}\label{lem:accepted_normalizer}
    Assume the input is $n$ perfect Bell pairs $\ket{\Phi^+}^{\otimes n}$ subjected to a Pauli error $E\in\Pauli_n$ on Bob before the biCEP circuit $U\otimes U^*$. For the standard biCEP protocol associated with $U$, equivalently for the stabilizer $S_U$ in Eq.~\eqref{eq:SU_def},
    \begin{align}
        \textnormal{acceptance} &\iff E\in N(S_U), \\
        \textnormal{accepted output is }\ket{\Phi^+} &\iff E\in S_U .
    \end{align}
    More generally, conditioned on acceptance, the logical Pauli error on the output pair is the coset $E S_U\in N(S_U)/S_U$.
\end{lemma}

\begin{proof}
    After applying $U\otimes U^*$, the effective Pauli error on Bob is $E_U = U^* E U^T$ according to the transpose trick. The standard biCEP measures $Z_{A,j}Z_{B,j}$ for $j=2,\dots,n$ and accepts iff all outcomes are $+1$. A Pauli error $E_U$ on Bob flips the $j$-th outcome iff $\{E_U,Z_j\}=0$. Hence all parity checks are accepted iff $E_U$ commutes with every generator of $S_0$, i.e. iff $E_U\in N(S_0)$.

    The above condition for acceptance is equivalent to $[E,U^T S_0 U^*]=0$. Meanwhile, we have $U^T Z_j U^* = (U^\dagger Z_j U)^* = g_j^*$, and $g_j^*$ differs from $g_j$ only by a possible sign. As we have assumed to ignore global phases for Pauli errors, this means precisely that $E\in N(S_U)$.

    Now assume the protocol succeeds, which means $E_U\in N(S_0)$. The remaining output pair is the logical qubit of the trivial code $S_0$, so the output is the target Bell state iff this logical Pauli is trivial, equivalently iff $E_U\in S_0$. Propagating back through the biCEP circuit, we have $E\in S_U$.
\end{proof}

We also describe an apparently different formulation of EPP from stabilizer codes.
\begin{definition}[Stabilizer-code EPP]
\label{def:postselected_stabilizer_epp}
Let $S\subseteq \mathcal P_n$ be the stabilizer group of an $[[n,1]]$ stabilizer code, with independent generators $g_1,\dots,g_{n-1}$. A stabilizer-code entanglement purification protocol associated with $S$ on $n$ noisy Bell pairs shared by Alice and Bob is as follows.

Alice measures stabilizers $g_{1,A},\dots,g_{n-1,A}$ on her $n$ qubits, and Bob measures the complex conjugate (in the computational basis) observables $g_{1,B}^*,\dots,g_{n-1,B}^*$ on his $n$ qubits. Equivalently, the protocol measures the bilocal stabilizers $g_{i,A}\otimes g_{i,B}^*$ for $i=1,\dots,n-1$. Alice and Bob communicate their measurement outcomes and accept the block iff all bilocal stabilizer outcomes are $+1$, corresponding to trivial syndrome. If any bilocal stabilizer outcome is $-1$, the block is discarded.

Upon success, Alice and Bob may apply local Clifford decoders
$U$ and $U^*$, respectively, satisfying $U S U^\dagger = \langle Z_2,\dots,Z_n\rangle$. They then discard the $n-1$ syndrome qubit pairs and keep the remaining Bell pair as the output.
\end{definition}

We can then see that the standard biCEP protocol using this Clifford $U$ is thus exactly the stabilizer-code EPP for $S$, as it effectively measures the bilocal stabilizers, accepts iff the syndrome is trivial, and outputs the logical Bell pair via discarding the $n-1$ measured qubit pairs. The decoding is ``built-in'' in the biCEP formulation as it is done before the measurements. Thus the standard biCEP formulation and the stabilizer-code EPP formulation are equivalent. The equivalence has an explicit correspondence: biCEP using Clifford $U$ is equivalent to stabilizer-code EPP with stabilizer $S_U = U^\dagger S_0 U$, where recall that $S_0$ is the trivial $[[n,1]]$ stabilizer.

\subsection{Success probability and successful output fidelity}
As a result of the biCEP formulation, for independent BDS input states, the biCEP success probability and successful output fidelity are
\begin{align}
    p_{\mathrm{succ}} &= \Pr(E\in N(S_U)), \\
    F_{\mathrm{succ}} &= \frac{\Pr(E\in S_U)}{\Pr(E\in N(S_U))} = \frac{\Pr(E\in S_U)}{p_{\mathrm{succ}}}.
\end{align}
From such probability statements we can in principle calculate these figures of merit from the probabilities of different Pauli error strings as
\begin{align}\label{eq:Fout_general}
    p_{\mathrm{succ}} &= \sum_{E\in N(S)}\ \prod_{i=1}^n p_i(E_i), \\
    F_{\mathrm{succ}} &= \frac{\sum_{E\in S}\ \prod_{i=1}^n p_i(E_i)}{\sum_{E\in N(S)}\ \prod_{i=1}^n p_i(E_i)}.
\end{align}

In fact, we can also express these figures of merit in terms of observables on the joint state of the BDS input. Consider any $P\in\mathcal P_n$ and we can define the bilocal Pauli observable $O(P) = P_A\otimes P_B^*$. Recall that for Pauli errors we ignore the global phases, and this bilocal observable itself is also explicitly independent of the global phase on $P$. Consider a standard $n$-to-1 biCEP protocol with Clifford $U$ with $S_U=\langle g_2,\dots,g_n\rangle$, where $g_j = U^\dagger Z_j U$.

Recall that the biCEP is successful when we measure $+1$ for all bilocal observables $O(g_i)$ with $g_i\in S_U$. Therefore, the biCEP success corresponds to a projector on the joint input state
\begin{align}
    \Pi_{\mathrm{succ}} = \prod_{j=2}^n \frac{I+O(g_j)}{2} = \frac{1}{|S_U|} \sum_{s\in S_U} O(s).
\end{align}
Therefore, for any joint input state $\rho_{AB}$ of the $n$ shared qubit pairs, the success probability can be expressed as
\begin{align}
    p_{\mathrm{succ}} = \Tr \left[\Pi_{\mathrm{succ}}\rho_{AB}\right] = \frac{1}{|S_U|} \sum_{s\in S_U} \Tr \left[ O(s)\rho_{AB} \right].
\end{align}

Next we consider the numerator of the successful output fidelity, $p_{\mathrm{succ}}F_{\mathrm{succ}}$. Consider the joint state after the biCEP circuit $\rho_{AB}^{(U)} = (U_A\otimes U_B^*)\rho_{AB} (U_A^\dagger\otimes U_B^T)$. This is effectively a reference frame transform from the initial joint state $\rho_{AB}$. The successful projector in this frame is
\begin{align}
    \Pi_0 = \prod_{j=2}^n \frac{I+Z_{A,j}Z_{B,j}}{2}.
\end{align}
This projector $\Pi_0$ is supported only on the measured qubits after the biCEP circuit, and the first qubit pair is untouched. Thus the numerator of the successful output fidelity can be expressed as
\begin{align}
    p_{\mathrm{succ}}F_{\mathrm{succ}} = \Tr \left[ (\Phi^+_1\otimes \Pi_0)\rho_{AB}^{(U)} \right],
\end{align}
where $\Phi^+_1$ is the projector on the target Bell state for the output state (first qubit pair). Expanding $\Phi^+_1$ and $\Pi_0$ gives
\begin{align}
    \Phi^+_1\otimes\Pi_0 = \frac{1}{|N(S_0)|} \sum_{h\in N(S_0)} O(h).
\end{align}
Propagating this observable back through the biCEP circuit $U\otimes U^*$ maps $N(S_0)$ to $N(S_U)$. Hence we have
\begin{align}
    p_{\mathrm{succ}}F_{\mathrm{succ}} = \frac{1}{|N(S_U)|} \sum_{h\in N(S_U)} \Tr \left[ O(h)\rho_{AB} \right].
\end{align}

For independent Bell-diagonal input pairs $\rho_{AB} = \bigotimes_{i=1}^n \rho_i$ with $\rho_i = \sum_{P\in \mathcal{P}_1} p_i(P) (I\otimes P)\Phi^+(I\otimes P)^\dagger$, define the single-pair correlations
\begin{align}
    \lambda_i(P) \coloneqq \Tr \left[ (P_A\otimes P_B^*)\rho_i \right].
\end{align}
Explicitly we have
\begin{align}
    \lambda_i(I) &= 1, \\
    \lambda_i(X) &= p_i(I)+p_i(X)-p_i(Y)-p_i(Z),\\
    \lambda_i(Y) &= p_i(I)-p_i(X)+p_i(Y)-p_i(Z),\\
    \lambda_i(Z) &= p_i(I)-p_i(X)-p_i(Y)+p_i(Z).
\end{align}
For a Pauli string $P=P_1\otimes\dots\otimes P_n$ we have
\begin{align}
    \Tr \left[ O(P)\rho_{AB} \right] = \prod_{i=1}^n \lambda_i(P_i).
\end{align}
Therefore,
\begin{align}
    p_{\mathrm{succ}} =& \frac{1}{|S_U|} \sum_{s\in S_U} \prod_{i=1}^n \lambda_i(s_i), \label{eq:bicep_bds_succ_prob} \\
    p_{\mathrm{succ}}F_{\mathrm{succ}} =& \frac{1}{|N(S_U)|} \sum_{h\in N(S_U)} \prod_{i=1}^n \lambda_i(h_i). \label{eq:bicep_bds_succ_fid_numer}
\end{align}
For an $[[n,1]]$ stabilizer code, $|S_U|=2^{n-1}$ and $|N(S_U)|=2^{n+1}$. For independent Werner inputs with visibilities $w_i$, the single-pair correlations reduce to $\lambda_i(I)=1$ and $\lambda_i(X)=\lambda_i(Y)=\lambda_i(Z)=w_i$. Thus when input states are all Werner states, $p_{\mathrm{succ}}$ and $p_{\mathrm{succ}}F_{\mathrm{succ}}$ will become polynomials of Werner parameters (visibilities) $w_i$ with non-negative coefficients, as long as the practical assumption $p_i(I)\ge 1/2$, i.e. each Werner state has Bell fidelity $F_i\ge 1/2$ and is thus entangled, holds.

\subsection{Examples}
\subsubsection{Bilocal CNOT EPP}
Consider the standard 2-to-1 bilocal CNOT protocol. Alice and Bob apply CNOT with pair 1 qubits as control and pair 2 qubits as target, measure the second pair in the computational basis. The EPP is successful if and only if the two measurement outcomes agree, and keep the first pair. For input states to this EPP, let $\nu_j^{(i)}$ denote the population of Bell component $j$ in input state $\rho_i$, where $i=1,2$ labels the input and $(\Phi^+,\Phi^-,\Psi^+,\Psi^-)$ correspond to $j=1,2,3,4$. Thus $\rho_i=\nu_1^{(i)}\Phi^+ + \nu_2^{(i)}\Phi^- + \nu_3^{(i)}\Psi^+ + \nu_4^{(i)}\Psi^-$, with $\sum_{j=1}^4\nu_j^{(i)}=1$ and $\nu_1^{(i)}=F_i$. This notation distinguishes Bell populations $\nu_j^{(i)}$ from the Bell correlations $\lambda_i(P)$ defined above. The successful output state is described by the transformation rules~\cite{deutsch1996quantum}
\begin{align}
    \nu_{1,\mathrm{succ}} &= \frac{\nu_1^{(1)}\nu_1^{(2)} + \nu_2^{(1)}\nu_2^{(2)}}{p_{\mathrm{succ}}} = F_{\mathrm{succ}},\\
    \nu_{2,\mathrm{succ}} &= \frac{\nu_1^{(1)}\nu_2^{(2)} + \nu_2^{(1)}\nu_1^{(2)}}{p_{\mathrm{succ}}},\\
    \nu_{3,\mathrm{succ}} &= \frac{\nu_3^{(1)}\nu_3^{(2)} + \nu_4^{(1)}\nu_4^{(2)}}{p_{\mathrm{succ}}},\\
    \nu_{4,\mathrm{succ}} &= \frac{\nu_3^{(1)}\nu_4^{(2)} + \nu_4^{(1)}\nu_3^{(2)}}{p_{\mathrm{succ}}},
\end{align}
with success probability 
\begin{align}
    p_{\mathrm{succ}} = (\nu_1^{(1)} + \nu_2^{(1)})(\nu_1^{(2)} + \nu_2^{(2)}) + (\nu_3^{(1)} + \nu_4^{(1)})(\nu_3^{(2)} + \nu_4^{(2)}).
\end{align}
One can see that the above expressions are invariant under exchange of $\rho_1$ and $\rho_2$. Werner input states are just a special case of BDS input, and for 2 Werner states with fidelities $F_1$ and $F_2$ we explicitly have
\begin{align}
    p_{\mathrm{succ}}^{\mathrm{CNOT}}(F_1,F_2) &= F_1 F_2 + \frac{F_1 (1-F_2) + F_2 (1-F_1)}{3} + \frac{5(1-F_1)(1-F_2)}{9},\\
    F_{\mathrm{succ}}^{\mathrm{CNOT}}(F_1,F_2) &= \frac{F_1 F_2 + (1-F_1)(1-F_2)/9}{p_{\mathrm{succ}}^{\mathrm{CNOT}}(F_1,F_2)} .
\end{align}

Then we demonstrate the analysis using the stabilizer framework. The associated stabilizer of the bilocal CNOT EPP is $S_{\mathrm{CNOT}} = \langle Z_1Z_2\rangle = \{II,ZZ\}$, with $|S_{\mathrm{CNOT}}|=2$. The normalizer consists of the Pauli strings commuting with $Z_1Z_2$: $N(S_{\mathrm{CNOT}}) = \{II,IZ,ZI,ZZ,XX,XY,YX,YY\}$ with $|N(S_{\mathrm{CNOT}})|=8$. For two generally different Werner inputs with visibilities $w_1,w_2$, then according to Eqs.~\eqref{eq:bicep_bds_succ_prob} and~\eqref{eq:bicep_bds_succ_fid_numer} we have
\begin{align}
    p_{\mathrm{succ}}^{\mathrm{CNOT}}(w_1,w_2) &= \frac{1+w_1w_2}{2},\\
    p_{\mathrm{succ}}^{\mathrm{CNOT}}(w_1,w_2)F_{\mathrm{succ}}^{\mathrm{CNOT}}(w_1,w_2) &= \frac{1+w_1+w_2+5w_1w_2}{8}.
\end{align}
One can easily verify that the above expressions in the visibility parameterization recover the above familiar expressions in the fidelity parameterization. Also, the expressions are invariant under the permutation of $w_1, w_2$, as expected.

\subsubsection{\texorpdfstring{$[[5,1,3]]$-code EPP}{Five-qubit-code EPP}}
Consider the standard $[[5,1,3]]$ stabilizer code $S_{5q}$ with generators $g_1=XZZXI, g_2=IXZZX, g_3=XIXZZ, g_4=ZXIXZ$ of size $|S_{5q}|=2^4=16$. This gives a 5-to-1 biCEP protocol. A convenient choice of logical gates is $\overline{X}=XXXXX, \overline{Z}=ZZZZZ$, which will facilitate the enumeration of $|N(S_{5q})|=2^6=64$ normalizers. 

We now list the stabilizer and normalizer elements explicitly, grouped by support. For instance, $XZZXI$ denotes a Pauli string with non-identity support $\{1,2,3,4\}$.

The stabilizer elements are
\begin{align*}
\begin{aligned}
    S_{5q}=\{&IIIII,XZZXI,YXXYI,ZYYZI,XXYIY,YYZIZ,ZZXIX,XYIYX,\\
    &YZIZY,ZXIXZ,XIXZZ,YIYXX,ZIZYY,IXZZX,IYXXY,IZYYZ\}.
\end{aligned}
\end{align*}
Grouped by support, the non-identity stabilizer elements are
\begin{align*}
\begin{array}{c|c}
    \text{support} & \text{stabilizer elements} \\
    \hline
    1234 & XZZXI,~ YXXYI,~ ZYYZI \\
    1235 & XXYIY,~ YYZIZ,~ ZZXIX \\
    1245 & XYIYX,~ YZIZY,~ ZXIXZ \\
    1345 & XIXZZ,~ YIYXX,~ ZIZYY \\
    2345 & IXZZX,~ IYXXY,~ IZYYZ
\end{array}
\end{align*}
Thus every 4-qubit support occurs exactly three times in $S_{5q}$.

For the normalizer, the only weight-0 element is the identity $IIIII$. The weight-3 normalizer elements are
\begin{align*}
\begin{array}{c|c}
    \text{support} & \text{normalizer elements} \\
    \hline
    123 & XYXII,~ YZYII,~ ZXZII \\
    124 & XXIZI,~ YYIXI,~ ZZIYI \\
    125 & XZIIZ,~ YXIIX,~ ZYIIY \\
    134 & XIYYI,~ YIZZI,~ ZIXXI \\
    135 & XIZIX,~ YIXIY,~ ZIYIZ \\
    145 & XIIXY,~ YIIYZ,~ ZIIZX \\
    234 & IXYXI,~ IYZYI,~ IZXZI \\
    235 & IXXIZ,~ IYYIX,~ IZZIY \\
    245 & IXIYY,~ IYIZZ,~ IZIXX \\
    345 & IIXYX,~ IIYZY,~ IIZXZ
\end{array}
\end{align*}
Hence every 3-qubit support occurs exactly three times in $N(S_{5q})$. The weight-4 normalizer elements are precisely the non-identity stabilizer elements listed above, so every 4-qubit support also occurs exactly three times in $N(S_{5q})$. The weight-5 normalizer elements are
\begin{align*}
\begin{aligned}
    \{& XXXXX, XXZYZ, XYYXZ, XYZZY, XZXYY, XZYZX,\\
    &YXYZZ, YXZXY, YYXZX, YYYYY, YZXXZ, YZZYX,\\
    &ZXXZY, ZXYYX, ZYXYZ, ZYZXX, ZZYXY, ZZZZZ \}.
\end{aligned}
\end{align*}
There are eighteen such full-support elements.

The key observation is that, the number of elements in $S_{5q}$ and $N(S_{5q})$ with a certain support depends only on the support size, not on which qubits are in the support. Therefore, for non-identical Werner visibilities $w_1,\dots,w_5$, the expressions for biCEP figures of merit as linear combinations of $w_i$ polynomials (the success probability and the success-weighted output fidelity) will reduce to linear combinations of elementary symmetric polynomials defined as
\begin{align}
    e_t(w_1,\dots,w_5) \coloneqq \sum_{\substack{J\subseteq\{1,2,3,4,5\}\\ |J|=t}} \prod_{i\in J}w_i,
\end{align}
with $e_0(w_1,\dots,w_5)=1$. More explicitly we have
\begin{align}
    \sum_{s\in S_{5q}} \prod_{i:s_i\neq I}w_i &= 1+3e_4(w_1,\dots,w_5),\\
    \sum_{h\in N(S_{5q})} \prod_{i:h_i\neq I}w_i &= 1 + 3 e_3(w_1,\dots,w_5) + 3 e_4(w_1,\dots,w_5) + 18 e_5(w_1,\dots,w_5).
\end{align}
Correspondingly, the biCEP figures of merit can be written as
\begin{align}
    p_{\mathrm{succ}}^{5q}(w_1,\dots,w_5) &= \frac{1+3e_4(w_1,\dots,w_5)}{16}, \\
    p_{\mathrm{succ}}^{5q}(w_1,\dots,w_5) F_{\mathrm{succ}}^{5q}(w_1,\dots,w_5) &= \frac{1 + 3e_3(w_1,\dots,w_5) + 3e_4(w_1,\dots,w_5) + 18e_5(w_1,\dots,w_5)}{64}.
\end{align}

\section{Guaranteed improvement of biCEP with Werner input by accumulating and shuffling}\label{sec:biCEP_AS}
In this section we give details of the guaranteed improvement of the success probability and the success-weighted output fidelity (the numerator of the successful output fidelity) for arbitrary $n$-to-1 biCEPs, when we use the accumulating and shuffling (AS) strategy, for arbitrary $n$ Werner entanglement sources. Recall that we have assumed uniform random permutation of the $n$ input states $\rho_1,\rho_2,\dots,\rho_n$ in the case without AS, according to the assumed unavailability of source labels. According to the preceding biCEP section, we know that for $n$ Werner states $\rho_i$ with visibilities $w_i$ in a certain order (permutation) $\pi(\cdot)$ as input $\bigotimes_{i=1}^n\rho_{\pi(i)}$, the analytical expressions for success probability and success-weighted output fidelity are multi-affine (due to the constant terms) polynomials of $w=(w_1,w_2,\dots,w_n)$. We may denote the success probability as $p_{\mathrm{succ}}(w_1,w_2,\dots,w_n)$, and the success-weighted output fidelity as $\tilde{F}_{\mathrm{succ}}(w_1,w_2,\dots,w_n)$. Then under the assumed operational random permutation, we have
\begin{align}
    \bar{p}_{\mathrm{succ}}(w) &= \frac{1}{|S_n|}\sum_{\pi\in S_n}p_{\mathrm{succ}}(w_{\pi(1)},w_{\pi(2)},\dots,w_{\pi(n)}), \\
    \bar{\tilde{F}}_{\mathrm{succ}}(w_1,w_2,\dots,w_n) &= \frac{1}{|S_n|}\sum_{\pi\in S_n}\tilde{F}_{\mathrm{succ}}(w_{\pi(1)},w_{\pi(2)},\dots,w_{\pi(n)}).
\end{align}
As finite linear combinations of multi-affine polynomials, they remain multi-affine polynomials. Moreover, they are obviously permutation-invariant, so they can be uniquely expressed as linear combinations of elementary symmetric polynomials
\begin{align}
    \bar{p}_{\mathrm{succ}}(w_1,w_2,\dots,w_n) &= \sum_{r=0}^n a_r e_r(w_1,w_2,\dots,w_n), \\
    \bar{\tilde{F}}_{\mathrm{succ}}(w_1,w_2,\dots,w_n) &= \sum_{r=0}^n b_r e_r(w_1,w_2,\dots,w_n).
\end{align}
In the following, for simplicity we will omit the overline in the notations of symmetrized figures of merit, as the random permutation has been assumed.

\subsection{Asymptotic limit of infinite rounds of entanglement distribution}
In the asymptotic limit of infinite rounds of entanglement distribution, the effective single-package state derived in the effective-state section is $\rho_{\mathrm{eff}}^{(\infty)} = \bar{\rho}^{\otimes n}$ with $\bar{\rho} = \sum_{i=1}^n\rho_i/n$. In other words, the effective single-package state is a tensor product of the expected state. In the following, we will prove the guaranteed improvement via properties of Schur-concave functions.

\begin{definition}[Vector majorization]
    For vector $x\in \R^n$, $x^\downarrow$ denotes the vector obtained by sorting the coordinates of $x$ in non-increasing order. Given $x,y\in \R^n$, we say that $x$ is majorized by $y$, and write $x\prec y$, if $\sum_{i=1}^m x_i^\downarrow \leq \sum_{i=1}^m y_i^\downarrow$ for $m=1,\dots,n-1$, and $\sum_{i=1}^n x_i = \sum_{i=1}^n y_i$.
\end{definition}
\begin{definition}[Schur-concave function]
    A function $f: D\to \R$, where $D\subseteq \R^n$, is Schur-concave on $D$ if $x\prec y, x,y\in D \Rightarrow f(x)\geq f(y)$. 
\end{definition}
A Schur-concave function is larger on the more evenly distributed vector, subject to a fixed coordinate sum. We use the following standard derivative criterion for Schur concavity

\begin{lemma}[Schur--Ostrowski criterion~\cite{marshall1979inequalities}]
    Let $D\subseteq \R^n$ be convex and invariant under coordinate permutations. Let $f\colon D\to \R$ be symmetric (permutation-invariant) and continuously differentiable on a neighborhood of $D$. If
    \begin{align}\label{eq:schur-ostrowski}
        (x_i-x_j) \left( \frac{\partial f(x)}{\partial x_i} - \frac{\partial f(x)}{\partial x_j} \right) \leq 0,
    \end{align}
    for every $x\in D$ and every pair $i,j$, then $f$ is Schur-concave on $D$.
\end{lemma}

Then our main objective is to analyze the Schur concavity of biCEP figures of merit expressions which can be written as linear combinations of elementary symmetric polynomials.

\begin{lemma}[Derivative identity]
    For $2\leq k\leq n$ and $i\neq j$ we have
    \begin{align}\label{eq:derivative-identity}
        \frac{\partial e_k(x)}{\partial x_i} - \frac{\partial e_k(x)}{\partial x_j} = -(x_i-x_j)e_{k-2}(x_{\widehat{i,j}}),
    \end{align}
    where $x_{\widehat{i,j}}$ denotes the $(n-2)$-dimensional vector obtained by removing the $i$-th and $j$-th coordinates from the original $n$-dimensional vector $x$. For $k=0,1$, the same difference is identically zero.
\end{lemma}

\begin{proof}
    The derivative $\partial e_k/\partial x_i$ is the sum of all square-free monomials of degree $k-1$ that do not use the coordinate $i$, i.e.
    \begin{align}
        \frac{\partial e_k(x)}{\partial x_i} = e_{k-1}(x_{\widehat{i}}).
    \end{align}
    Then we separate the terms according to whether they contain $x_j$. For $i\neq j$ we have
    \begin{align}
        e_{k-1}(x_{\widehat{i}}) &= e_{k-1}(x_{\widehat{i,j}}) + x_j e_{k-2}(x_{\widehat{i,j}}),\\
        e_{k-1}(x_{\widehat{j}}) &= e_{k-1}(x_{\widehat{i,j}}) + x_i e_{k-2}(x_{\widehat{i,j}}).
    \end{align}
    Direct subtraction between the above two expressions gives Eq.~\eqref{eq:derivative-identity}. The cases for $k=0,1$ are obvious.
\end{proof}
Then we have the Schur concavity for linear combinations of elementary symmetric polynomials with non-negative coefficients
\begin{theorem}
    Consider polynomial $p(x) = \sum_{k=0}^n a_k e_k(x)$ with $a_k\geq 0$ for all $k$. Then $p$ is Schur-concave on the non-negative orthant $\R_+^n$.
\end{theorem}
\begin{proof}
    The polynomial $p(x)$ is symmetric (permutation-invariant) and continuously differentiable, and we focus on the convex and coordinate-permutation-invariant domain $x\in \R_+^n$. By the derivative identity above we have
    \begin{align}
        (x_i-x_j) \left( \frac{\partial e_k}{\partial x_i}(x) - \frac{\partial e_k}{\partial x_j}(x) \right) = - (x_i-x_j)^2 e_{k-2}(x_{\widehat{i,j}})
    \end{align}
    for $i\ne j$ and $k\geq 2$, while the same expression is 0 for $k=0,1$. Since $x\in \R_+^n$, every elementary symmetric polynomial in the remaining coordinates is non-negative $e_{k-2}(x_{\widehat{i,j}})\geq 0$. Therefore, we have 
    \begin{align}
        (x_i-x_j) \left( \frac{\partial e_k}{\partial x_i}(x) - \frac{\partial e_k}{\partial x_j}(x) \right) \leq 0,~ \forall k.
    \end{align}
    According to linearity, we can multiply the above expression by $a_k\geq 0$ for each $e_k$ and sum over $k$. This yields
    \begin{align}
        (x_i-x_j) \left( \frac{\partial p}{\partial x_i}(x) - \frac{\partial p}{\partial x_j}(x) \right) \leq 0.
    \end{align}
    The Schur--Ostrowski criterion then gives that $p(x)$ is Schur-concave on $\R_+^n$.
\end{proof}

Using the above results, we have the AS enhancement in the asymptotic limit of $m\to\infty$.

\begin{remark}
    The successful output fidelity (not weighted by the success probability) is actually not guaranteed to be improved by the AS strategy for arbitrary biCEP. To demonstrate this argument we can consider the $[[5,1,3]]$-code biCEP. Then we can consider some numerical example, such as $(w_1,\dots,w_5) = (0.9,0.9,0.9,0.6,0.6)$ with $\bar{w}=0.78$. Direct calculation shows that without AS $F_{\mathrm{succ}}^{5q}(0.9,0.9,0.9,0.6,0.6)\approx 0.9931$, while with AS $F_{\mathrm{succ}}^{5q}(0.78,0.78,0.78,0.78,0.78)\approx 0.9915$. Nevertheless, this does not mean that the AS strategy never has any advantage for successful output fidelity either. More fairly speaking, the AS strategy can enhance successful output fidelity for some but not all combinations of Werner input states, while the advantageous region depends on the biCEP protocol.
\end{remark}

\subsection{Finite rounds of entanglement distribution}

Recall that for arbitrary biCEP, the expressions of both success probability $p_{\mathrm{succ}}(w_1,\dots,w_n)$ and the success-weighted output fidelity $\tilde{F}_{\mathrm{succ}}(w_1,\dots,w_n)$ under the visibility parameterization have the form
\begin{equation}\label{eq:generic-multiaffine}
  f(w_1,\dots,w_n) = \sum_{T\subseteq [n]} c_T\prod_{j\in T}w_j,
\end{equation}
where $[n]=\{1,\dots,n\}$ and $c_T\geq 0$. It is then sufficient to prove the AS comparison for an arbitrary function $f(w_1,\dots,w_n)$ in the form of Eq.~\eqref{eq:generic-multiaffine}.  

For the operationally randomized no-AS baseline, we have
\begin{align}
    f_{\mathrm{noAS}}(w_1,\dots,w_n) = \frac{1}{n!} \sum_{\pi\in S_n} f(w_{\pi(1)},\dots,w_{\pi(n)}).
\end{align}
For any fixed subset $T\subseteq[n]$ with $|T|=r$, the random set $\{\pi(j):j\in T\}$ is uniformly distributed over all $r$-element subsets of $[n]$. Meanwhile, there are $r!(n-r)!$ permutations corresponding to one same $\prod_{j\in T}w_{\pi(j)}$. Therefore,
\begin{align}\label{eq:f-noas}
    f_{\mathrm{noAS}}(w_1,\dots,w_n) &= \sum_{T\subseteq[n]} c_T \left(\frac{1}{n!} \sum_{\pi\in S_n} \prod_{j\in T}w_{\pi(j)} \right) \notag\\
    &= \sum_{T\subseteq[n]} c_T \frac{e_{|T|}(w_1,\dots,w_n)}{\binom n{|T|}}.
\end{align}

With AS, consider the length-$mn$ list of Werner parameters after $m$ distribution rounds
\begin{align}
    w^{(m)} = (\underbrace{w_1,\dots,w_1}_{m\ \mathrm{times}}, \underbrace{w_2,\dots,w_2}_{m\ \mathrm{times}}, \dots, \underbrace{w_n,\dots,w_n}_{m\ \mathrm{times}})
\end{align}
A single AS package is an ordered sample $Y=(Y_1,\dots,Y_n)$ drawn without replacement from this list. As AS uniformly shuffle the accumulated states, in one package there is no special slot, which means that the random variable sequence $Y$ is exchangeable. Then for a fixed $T\subseteq[n]$ with $|T|=r$, exchangeability of the ordered sample gives
\begin{align}
    \E_{\mathrm{AS}} \left[\prod_{j\in T}Y_j\right] = \frac{e_r(w^{(m)})}{\binom{mn}{r}}.
\end{align}
In other words, the $r$ slots indexed by $T$ contain a uniformly random, ordered selection of $r$ distinct (without repetition) elements from the length-$mn$ list. Each $r$-element subset of positions in the list is equally likely, so the expected product is the normalized elementary symmetric mean of degree $r$. Therefore,
\begin{equation}\label{eq:f-as}
    f_{\mathrm{AS}}^{(m)} = \E_{\mathrm{AS}}[f(Y_1,\dots,Y_n)] = \sum_{T\subseteq[n]} c_T \frac{e_{|T|}(w^{(m)})}{\binom{mn}{|T|}}.
\end{equation}

The comparison between Eq.~\eqref{eq:f-noas} and Eq.~\eqref{eq:f-as} reduces to a single inequality for the mean elementary symmetric polynomials (elementary symmetric polynomials divided by the number of its terms), which we may also call elementary symmetric means.

\begin{lemma}\label{lem:esm-ineq}
Let $w_1,\dots,w_n\geq 0$, $m\geq 1$, and let $w^{(m)}$ be the length-$mn$ vector obtained by repeating each $w_i$ exactly $m$ times. Then $\forall r=0,1,\dots,n$,
\begin{equation}\label{eq:esm-ineq}
    \frac{e_r(w^{(m)})}{\binom{mn}{r}} \geq \frac{e_r(w_1,\dots,w_n)}{\binom nr}.
\end{equation}
\end{lemma}

\begin{proof}
The cases $r=0$ and $r=1$ hold with equality. Then we focus on $2\leq r\leq n$. Define the elementary symmetric means
\begin{align}
    A_k=\frac{e_k(w_1,\dots,w_n)}{\binom nk},
\end{align}
for $k=0,1,\dots,n$ with $A_0=1$. By Maclaurin's inequality for non-negative variables we have
\begin{align}
    A_k^{1/k}\geq A_r^{1/r},
\end{align}
for $1\leq k\leq r$. Equivalently,
\begin{align}
    A_k\geq A_r^{k/r},
\end{align}
for $0\leq k\leq r$, where the $k=0$ case is basically $1=1$.

Using the generating function for elementary symmetric polynomials,
\begin{align}
    \sum_{l=0}^{mn} e_l(w^{(m)})t^l &= \prod_{i=1}^{mn}(1+w_i t) \notag \\
    &= \left(\prod_{i=1}^n(1+w_i t)\right)^m = \prod_{i=1}^n(1+w_i t)^m \notag \\
    & = \left(\sum_{k=0}^n e_k(w_1,\dots,w_n)t^k\right)^m.
\end{align}
Taking the coefficient of $t^r$, we get
\begin{align}
    e_r(w^{(m)}) = \sum_{\substack{k_1+\dots+k_m=r\\0\leq k_l\leq n}} \prod_{l=1}^m e_{k_l}(w_1,\dots,w_n).
\end{align}
Since $k_1+\dots+k_m=r$, we implicitly have $k_l\leq r$, and hence
\begin{align}
    \prod_{l=1}^m e_{k_l}(w) = \prod_{l=1}^m \binom n{k_l}A_{k_l} \geq \prod_{l=1}^m \binom n{k_l} A_r^{k_l/r} = A_r \prod_{l=1}^m \binom n{k_l},
\end{align}
where for the inequality we have used the result of Maclaurin's inequality. Summing over all $m$-tuples $(k_1,\dots,k_m)$ s.t. $k_1+\dots+k_m=r$, we have
\begin{align}
    e_r(w^{(m)}) \geq A_r \sum_{k_1+\dots+k_m=r} \prod_{l=1}^m \binom n{k_l}.
\end{align}
By repeated application of Vandermonde's identity, the sum on the right hand side of the above inequality is $\binom{mn}{r}$. Therefore, we conclude with
\begin{align}
    e_r(w^{(m)}) \geq A_r\binom{mn}{r} = \frac{e_r(w_1,\dots,w_n)}{\binom nr}\binom{mn}{r},
\end{align}
which is exactly Eq.~\eqref{eq:esm-ineq}.
\end{proof}

Using the above results, we have the AS enhancement theorem for arbitrary finite number of entanglement distribution rounds.

\subsection{Monotonicity with the number of accumulation rounds}
Lemma~\ref{lem:esm-ineq} compares AS after any finite number of accumulation rounds with the non-AS baseline, which corresponds to $m=1$. Here we strengthen the comparison to that the relevant elementary symmetric means are non-decreasing as $m$ increases.

\begin{lemma}\label{lem:esm-adjacent}
    Let $w_1,\dots,w_n\geq 0$, let $m\geq 2$, and let $w^{(m)}$ be the accumulated visibility vector defined above. Then
    \begin{align}
        \frac{e_r(w^{(m)})}{\binom{mn}{r}} \geq \frac{e_r(w^{(m-1)})}{\binom{(m-1)n}{r}} \label{eq:esm-adjacent}
    \end{align}
    for every $0\leq r\leq(m-1)n$. In particular, Eq.~\eqref{eq:esm-adjacent} holds for $r=0,\dots,n$, which is the range relevant to an $n$-input biCEP.
\end{lemma}

\begin{proof}
    For $q\geq 1$ and $0\leq r\leq qn$, consider an ordered uniform sample without replacement
    \begin{align}
        Y_1^{(q)},\dots,Y_r^{(q)}
    \end{align}
    from the length-$qn$ list $w^{(q)}$. By the same counting argument used in the derivation of Eq.~\eqref{eq:f-as},
    \begin{align}
        \E_q\left[\prod_{j=1}^rY_j^{(q)}\right] = \frac{e_r(w^{(q)})}{\binom{qn}{r}}. \label{eq:esm-adjacent-sampling}
    \end{align}
    It is easy to see that the cases $r=0$ and $r=1$ hold with equality. We may therefore assume $r\geq 2$ in the following. For every nonnegative number $x$, the layer cake representation gives
    \begin{align}
        x=\int_0^\infty \boldsymbol{1}_{\{t\leq x\}}\,dt,
    \end{align}
    where $\boldsymbol{1}_{\{t\leq x\}}$ denotes the indicator of the condition $t\leq x$, as it equals to $1$ when the condition holds and $0$ otherwise. Applying this representation to every $Y_j^{(q)}$, for each fixed realization of the ordered sample we have
    \begin{align}
        \prod_{j=1}^rY_j^{(q)}=\int_{[0,\infty)^r}\boldsymbol{1}_{\{Y_j^{(q)}\geq t_j,\forall j\}}\,dt_1\cdots dt_r.
    \end{align}
    The ordered sample has only $(qn)_r$ possible outcomes, so its expectation is a finite average. Exchanging this finite average with the integral gives
    \begin{align}\label{eq:esm-adjacent-layer cake}
        \frac{e_r(w^{(q)})}{\binom{qn}{r}}=\int_{[0,\infty)^r}\Pr_q\left(Y_j^{(q)}\geq t_j, \forall j\right)\,dt_1\cdots dt_r,
    \end{align}
    where in the above $\Pr_q$ and $\E_q$ denote probability and expectation, respectively, with respect to this $q$-dependent uniform sampling procedure.
    
    Then we analyze the integrand, and we fix the thresholds $t_1,\dots,t_r$ to be constants for now. Also note that permuting the ordered sample does not change the joint probability distribution, so permuting the threshold positions does not change the probability in Eq.~\eqref{eq:esm-adjacent-layer cake} either. Without loss of generality, we may assume
    \begin{align}
        t_1\geq t_2\geq\cdots\geq t_r.
    \end{align}
    For $j=1,\dots,r$, define
    \begin{align}
        B_j:=\{i\in[n]:w_i\geq t_j\}.
    \end{align}
    The ordering of the thresholds implies $B_1\subseteq B_2\subseteq\cdots\subseteq B_r$. Although values may repeat in $w^{(q)}$, regard its $qn$ occurrences as distinct labeled entries $(i,a)$, where $i\in[n]$ is the base index associated with the value $w_i$ and $a\in[q]$ distinguishes its $q$ occurrences. Let $N_j^{(q)}$ denote the number of ordered outcomes obtained by drawing $j$ distinct accumulated entangled pairs without replacement into sample slots $1,\dots,j$, subject to the first $j$ threshold requirements, with $N_0^{(q)}=1$. We derive a recurrence relation for $N_j^{(q)}$ in the following by extending valid selections one slot at a time.

    Fix $j\in\{1,\dots,r\}$ and consider a valid selection for the first $j-1$ slots. By definition of $B_l$, the entry in slot $l\leq j-1$ has a base index $i\in B_l$. Because $B_l\subseteq B_j$, every one of the $j-1$ previously selected entries has its base index in $B_j$. There are $s_j=|B_j|$ base indices in $B_j$, and each occurs $q$ times in $w^{(q)}$. Hence exactly $qs_j$ labeled entries satisfy the $j$-th threshold requirement. The existing valid selection has already used $j-1$ distinct entries among these $qs_j$ entries, leaving exactly $qs_j-(j-1)$ eligible entries for slot $j$. This number is the same for every valid selection of length $j-1$, and adding one of these remaining entries produces a unique valid selection of length $j$.
    
    If $N_{j-1}^{(q)}>0$, the existence of a valid selection implies $qs_j\geq j-1$, and the preceding extension count gives $N_j^{(q)}=N_{j-1}^{(q)}[qs_j-(j-1)]$. If $N_{j-1}^{(q)}=0$, then $N_j^{(q)}=0$, because deleting the last entry of any valid length-$j$ selection would produce a valid length-$(j-1)$ selection. The two cases can therefore be written together as
    \begin{align}
        N_j^{(q)}=N_{j-1}^{(q)}\max\{qs_j-(j-1),0\}. \label{eq:esm-adjacent-prefix-recurrence}
    \end{align}
    Starting from $N_0^{(q)}=1$ and applying Eq.~\eqref{eq:esm-adjacent-prefix-recurrence} successively gives
    \begin{align}\label{eq:esm-adjacent-valid-count}
        N_r^{(q)}=\prod_{j=1}^r\max\{qs_j-(j-1),0\}. 
    \end{align}
    
    Meanwhile, the total number of ordered selections of $r$ distinct entries from the $qn$ labeled entries of $w^{(q)}$ is
    \begin{align}
        (qn)_r=qn(qn-1)\cdots(qn-r+1)=\prod_{j=1}^r[qn-(j-1)]. \label{eq:esm-adjacent-total-count}
    \end{align}
    Since every ordered selection is equally likely, the desired threshold-event probability is the number $N_r^{(q)}$ of valid selections divided by the total number $(qn)_r$ of selections. Substituting Eqs.~\eqref{eq:esm-adjacent-valid-count} and~\eqref{eq:esm-adjacent-total-count} and factoring the quotient gives
    \begin{align}\label{eq:esm-adjacent-threshold-probability}
        \Pr_q\left(Y_j^{(q)}\geq t_j, \forall j\right) =\prod_{j=1}^r\frac{\max\{qs_j-(j-1),0\}}{qn-(j-1)}. 
    \end{align}
    
    We compare each factor in Eq.~\eqref{eq:esm-adjacent-threshold-probability} for $q=m$ and $q=m-1$. Fix $s\in\{0,\dots,n\}$ and $c\in\{0,\dots,r-1\}$. The assumption $r\leq(m-1)n$ guarantees that both denominators below are positive. If $(m-1)s-c\leq 0$, then the factor for $m-1$ vanishes and is no larger than the nonnegative factor for $m$. If $(m-1)s-c>0$, direct subtraction gives
    \begin{align}
        \frac{ms-c}{mn-c}-\frac{(m-1)s-c}{(m-1)n-c}=\frac{c(n-s)}{(mn-c)((m-1)n-c)}\geq 0. \label{eq:esm-adjacent-factor}
    \end{align}
    Thus every factor in the threshold-event probability is non-decreasing from $m-1$ to $m$, for every constant sequence of thresholds $t_1,\dots,t_r$
    \begin{align}
        \Pr_m\left(Y_j^{(m)}\geq t_j, \forall j\right)\geq\Pr_{m-1}\left(Y_j^{(m-1)}\geq t_j, \forall j\right).
    \end{align}
    Integrating the above inequality according to Eq.~\eqref{eq:esm-adjacent-layer cake} proves Eq.~\eqref{eq:esm-adjacent}.
\end{proof}
Then the monotonicity of $p_{\mathrm{succ}}$ or $\tilde{F}_{\mathrm{succ}}$ for biCEPs follows from Eq.~\eqref{eq:f-as}.

\section{Proofs for bilocal CNOT (DEJMPS) EPP results without memory decoherence}\label{sec:bilocal_cnot_epp_results}
In this section we prove results about the bilocal CNOT (DEJMPS) EPP, including an additional property that is not covered in the main text.

\subsection{Proof for Proposition~\ref{thm:succ_improvement}}
The Proposition's proof comes from main text Lemma~\ref{thm:eff_output} and the following Lemma about the analytical properties of the DEJMPS protocol.
\begin{lemma}\label{thm:DEJMPS_succ_properties}
    For the DEJMPS protocol with two Werner states as inputs, we have $F_{\mathrm{succ}}\left(\frac{F_1+F_2}{2},\frac{F_1+F_2}{2}\right) \geq F_{\mathrm{succ}}(F_1,F_2)$ and $p_{\mathrm{succ}}\left(\frac{F_1+F_2}{2},\frac{F_1+F_2}{2}\right) \geq p_{\mathrm{succ}}(F_1,F_2)$ for $\forall F_1,F_2\in[1/2,1]$, where $F_{\mathrm{succ}}(\cdot,\cdot)$ and $p_{\mathrm{succ}}(\cdot,\cdot)$ denote the successful output fidelity and success probability, respectively, as functions of the two input fidelities.
\end{lemma}
\begin{proof}
    Under the assumption of Werner states $\rho_1$ and $\rho_2$, we have successful output fidelity and success probability as function of two input fidelities:
    \begin{align}
        F_{\mathrm{succ}}(F_1,F_2) =& \frac{F_1F_2 + \frac{(1-F_1)(1-F_2)}{9}}{p_{\mathrm{succ}}(F_1,F_2)},\\
        p_{\mathrm{succ}}(F_1,F_2) =& F_1F_2 + \frac{F_1(1-F_2)+(1-F_1)F_2}{3} + \frac{5(1-F_1)(1-F_2)}{9}.
    \end{align}

    We first examine successful output fidelity by taking the difference between $F_{\mathrm{succ}}\left(\frac{F_1+F_2}{2},\frac{F_1+F_2}{2}\right)$ and $F_{\mathrm{succ}}(F_1,F_2)$:
    \begin{align}
        F_{\mathrm{succ}}\left(\frac{F_1+F_2}{2},\frac{F_1+F_2}{2}\right) - F_{\mathrm{succ}}(F_1,F_2) = \frac{3(F_1-F_2)^2(7-2F_1-2F_2)}{2(5-2F_1-2F_2+2F_1^2+2F_2^2+4F_1F_2)(5-2F_1-2F_2+8F_1F_2)},
    \end{align}
    whose numerator and denominator can be straightforwardly verified to be non-negative $\forall F_1,F_2\in[1/2,1]$.

    Then for success probability, we can similarly take the difference between $p_{\mathrm{succ}}\left(\frac{F_1+F_2}{2},\frac{F_1+F_2}{2}\right)$ and $p_{\mathrm{succ}}(F_1,F_2)$:
    \begin{align}
        p_{\mathrm{succ}}\left(\frac{F_1+F_2}{2},\frac{F_1+F_2}{2}\right) - p_{\mathrm{succ}}(F_1,F_2) = \frac{2(F_1-F_2)^2}{9},
    \end{align}
    which is obviously non-negative $\forall F_1,F_2\in[1/2,1]$.
\end{proof}
The proof for monotonicity with increasing $m$ is as follows.
\begin{proof}
    From main text Lemma~\ref{thm:eff_output} the success probability is
    \begin{align}
        p_{\mathrm{succ}}(\rho_{\mathrm{eff},n=2}^{(m)}) = \frac{2m-2}{2m-1}p_{\mathrm{succ}}\left(\frac{F_1+F_2}{2},\frac{F_1+F_2}{2}\right) + \frac{1}{2m-1}p_{\mathrm{succ}}(F_1,F_2),
    \end{align}
    as a convex combination of $p_{\mathrm{succ}}\left(\frac{F_1+F_2}{2},\frac{F_1+F_2}{2}\right)$ and $p_{\mathrm{succ}}(F_1,F_2)$. Meanwhile, $\frac{2m-2}{2m-1}$ monotonically increases with increasing $m$, thus $p_{\mathrm{succ}}(\rho_{\mathrm{eff},n=2}^{(m)})$ also monotonically increases with increasing $m$.

    The successful output fidelity is 
    \begin{align}
        F_{\mathrm{succ}}(\rho_{\mathrm{eff},n=2}^{(m)}) &= \frac{\frac{2m-2}{2m-1}p_{\mathrm{succ}}\left(\frac{F_1+F_2}{2},\frac{F_1+F_2}{2}\right)F_{\mathrm{succ}}\left(\frac{F_1+F_2}{2},\frac{F_1+F_2}{2}\right) + \frac{1}{2m-1}p_{\mathrm{succ}}(F_1,F_2)F_{\mathrm{succ}}(F_1,F_2)} {p_{\mathrm{succ}}(\rho_{\mathrm{eff},n=2}^{(m)})} \notag\\
        &= \beta_m F_{\mathrm{succ}}\left(\frac{F_1+F_2}{2},\frac{F_1+F_2}{2}\right) + (1-\beta_m)F_{\mathrm{succ}}(F_1,F_2),
    \end{align}
    where
    \begin{align}
        \beta_m = \frac{(2m-2)p_{\mathrm{succ}}\left(\frac{F_1+F_2}{2},\frac{F_1+F_2}{2}\right)}{(2m-2)p_{\mathrm{succ}}\left(\frac{F_1+F_2}{2},\frac{F_1+F_2}{2}\right) + p_{\mathrm{succ}}(F_1,F_2)}.
    \end{align}
    As $p_{\mathrm{succ}}\left(\frac{F_1+F_2}{2},\frac{F_1+F_2}{2}\right)$ and $p_{\mathrm{succ}}(F_1,F_2)$ are both non-negative, we have that $\beta_m$ monotonically increases with increasing $m$, thus $F_{\mathrm{succ}}(\rho_{\mathrm{eff},n=2}^{(m)})$ also monotonically increases with increasing $m$.
\end{proof}

\subsection{Proof for Proposition~\ref{thm:ais_bias}}
The Proposition comes directly from the following Lemma using the analytical properties of the DEJMPS protocol, given main text Lemma~\ref{thm:eff_output} and the convex mixture form of $\rho_{\mathrm{eff},n=2}^{(m)}$.
\begin{lemma}
    For Bell diagonal inputs, consider AIS $\bar{\rho}=(\bar{F},a(1-\bar{F}),b(1-\bar{F}),(1-a-b)(1-\bar{F}))$ with $a\in[0,1/3]$. Then for BDS input state $\rho_1=(F_1,a(1-F_1),b_1(1-F_1),(1-a-b_1)(1-F_1))$ together with $\rho_2 = 2\bar{\rho}-\rho_1$, we have $F_{\mathrm{succ}}(\bar{\rho}\otimes\bar{\rho}) \geq F_{\mathrm{succ}}(\rho_1\otimes\rho_2) = F_{\mathrm{succ}}(\rho_2\otimes\rho_1)$ for $\forall\bar{F}\in[1/2,1]$, if and only if $a\in[2-\sqrt{3},1/3]$.
\end{lemma}
\begin{proof}
    The sign of $F_{\mathrm{succ}}(\bar{\rho}\otimes\bar{\rho})-F_{\mathrm{succ}}(\rho_1\otimes\rho_2)$ is the same as the sign of $p_{\mathrm{succ}}(\bar{\rho}\otimes\bar{\rho})p_{\mathrm{succ}}(\rho_1\otimes\rho_2)[F_{\mathrm{succ}}(\bar{\rho}\otimes\bar{\rho})-F_{\mathrm{succ}}(\rho_1\otimes\rho_2)]$. Denote the latter quantity by $\Delta_{\mathrm{num}}F_{\mathrm{succ}}$; explicitly,
    \begin{align}
        \Delta_{\mathrm{num}}F_{\mathrm{succ}} =& p_{\mathrm{succ}}(\bar{\rho}\otimes\bar{\rho})p_{\mathrm{succ}}(\rho_1\otimes\rho_2) [F_{\mathrm{succ}}(\bar{\rho}\otimes\bar{\rho})-F_{\mathrm{succ}}(\rho_1\otimes\rho_2) ]\notag\\
        =& [a^2(1-\bar{F})^2 + \bar{F}^2 ] [1-2\bar{F}+2a(1-\bar{F})(2F_1-2a_1F_1+2a_1-1) + 2(F_1-a_1F_1 + a_1)(2\bar{F}+a_1F_1-F_1-a_1) ]\notag\\
        &- [(1-a)^2(1-\bar{F})^2+(a+\bar{F}-a\bar{F})^2 ] [a_1(2a-2a\bar{F}-a_1+a_1F_1)(1-F_1) + (2\bar{F}-F_1)F_1 ].
    \end{align}

    When $a_1=a$, we can simplify $\Delta_{\mathrm{num}}F_{\mathrm{succ}}$ as:
    \begin{align}
        \left.\Delta_{\mathrm{num}}F_{\mathrm{succ}}\right\vert_{a_1=a} = [(1-2a+a^2+2a^3) - 2(1-3a+a^2+a^3)\bar{F}](\bar{F}-F_1)^2,
    \end{align}
    from which we see that we can focus on the first bracket to determine the sign of the difference.

    Now we take the partial derivative of $\left.\Delta_{\mathrm{num}}F_{\mathrm{succ}}\right\vert_{a_1=a}$ with respect to $\bar{F}$ as:
    \begin{align}
        \frac{\partial}{\partial\bar{F}}\left.\Delta_{\mathrm{num}}F_{\mathrm{succ}}\right\vert_{a_1=a} = - 2(1-3a+a^2+a^3),
    \end{align}
    which is negative for $a\in[0,1/3]$. Therefore, we only need to consider the case when $\bar{F}=1$:
    \begin{align}
        \left.\Delta_{\mathrm{num}}F_{\mathrm{succ}}\right\vert_{a_1=a,\bar{F}=1} = -1 + 4a - a^2,
    \end{align}
    which is non-negative if and only if $a\in[2-\sqrt{3},1/3]$.
\end{proof}

\subsection{Proof for Proposition~\ref{thm:arb_bds_succ_prob_norm_fid}}
The Proposition comes directly from the following Lemma using the analytical properties of the DEJMPS protocol, given main text Lemma~\ref{thm:eff_output} and the convex mixture form of $\rho_{\mathrm{eff},n=2}^{(m)}$.
\begin{lemma}\label{thm:arb_bds_fid_and_normfid}
    For arbitrary BDS AIS $\bar{\rho}$ of two BDS input states $\rho_1$ and $\rho_2$ to the DEJMPS protocol, we have $p_{\mathrm{succ}}(\bar{\rho}\otimes\bar{\rho}) \geq p_{\mathrm{succ}}(\rho_1\otimes\rho_2) = p_{\mathrm{succ}}(\rho_2\otimes\rho_1)$ and $\tilde{F}_{\mathrm{succ}}(\bar{\rho}\otimes\bar{\rho}) \geq \tilde{F}_{\mathrm{succ}}(\rho_1\otimes\rho_2) = \tilde{F}_{\mathrm{succ}}(\rho_2\otimes\rho_1)$.
\end{lemma}
\begin{proof}
     Consider $\bar{\rho}=(\bar{F},a(1-\bar{F}),b(1-\bar{F}),(1-a-b)(1-\bar{F}))$ with $a\in[0,1/3]$. We parameterize $\rho_1=(F_1,a_1(1-F_1),b_1(1-F_1),(1-a_1-b_1)(1-F_1))$, then $\rho_2=(\nu_1^{(2)},\nu_2^{(2)},\nu_3^{(2)},\nu_4^{(2)})$ is uniquely determined given a fixed $\bar{\rho}$. Note that positivity of density matrix requires that all diagonal elements be non-negative, which exerts additional constraints on the parameterization, but for now we proceed without worrying about the constraints and derive the following analytical expressions for success probability and output fidelity.
     \begin{align}
         p_{\mathrm{succ}}(\bar{\rho}\otimes\bar{\rho}) =& (1-a)^2(1-\bar{F})^2 + (a+\bar{F}-a\bar{F})^2,\\
         p_{\mathrm{succ}}(\rho_1\otimes\rho_2) =& 2a(1-\bar{F})[2F_1+2a_1(1-F_1)-1] + (1 - 2\bar{F}) \notag\\
         & + 2[a_1(1-F_1)+F_1][2\bar{F}-a_1(1-F_1)-F_1],\\
         F_{\mathrm{succ}}(\bar{\rho}\otimes\bar{\rho}) =& \frac{a^2(1-\bar{F})^2+\bar{F}^2}{p_{\mathrm{succ}}(\bar{\rho}\otimes\bar{\rho})},\\
         F_{\mathrm{succ}}(\rho_2\otimes\rho_1) =& \frac{2a a_1(1-\bar{F})(1-F_1) - a_1^2(1-F_1)^2 + (2\bar{F}-F_1)F_1}{p_{\mathrm{succ}}(\rho_1\otimes\rho_2)}.
     \end{align}
     Then straightforward evaluations result in
     \begin{align}
         p_{\mathrm{succ}}(\bar{\rho}\otimes\bar{\rho}) - p_{\mathrm{succ}}(\rho_1\otimes\rho_2) = 2[a_1(1-F_1)-a(1-\bar{F}) + F_1 - \bar{F}]^2 \geq 0,
     \end{align}
     and
     \begin{align}
         \tilde{F}_{\mathrm{succ}}(\bar{\rho}\otimes\bar{\rho}) - \tilde{F}_{\mathrm{succ}}(\rho_2\otimes\rho_1) = [a(1-\bar{F})-a_1(1-F_1)]^2 + (\bar{F}-F_1)^2 \geq 0.
     \end{align}
     The non-negativity of the differences proves the statement.
\end{proof}

\subsection{Almost concavity of DEJMPS output fidelity with Werner input states}
We have shown that $F_{\mathrm{succ}}(\frac{F_1+F_2}{2},\frac{F_1+F_2}{2}) \geq F_{\mathrm{succ}}(F_1,F_2)$. This seems to imply that $F_{\mathrm{succ}}(F_1,F_2)$ is a concave function for $(F_1,F_2)\in[1/2,1]\times[1/2,1]$. Therefore, here we study the concavity of $F_{\mathrm{succ}}(F_1,F_2)$, and show that it is \textit{almost} concave on $[1/2,1]\times[1/2,1]$, with exception only around $(1/2,1/2)$.
\begin{proposition}
    The successful output fidelity of the DEJMPS protocol given two Werner input states with fidelities $F_1$ and $F_2$, namely $F_{\mathrm{succ}}(F_1,F_2)$, is concave for $F_1,F_2\in[0.54008,1]$.
\end{proposition}
\begin{proof}
    First of all, recall that for any twice-differentiable multivariate function, it is concave on a convex set if and only if its Hessian is negative semidefinite on the same convex set. For the bivariate $F_{\mathrm{succ}}(F_1,F_2)$, we can derive its second-order partial derivatives as follows:
    \begin{align}
        (F_{\mathrm{succ}})_{F_1F_1} =& \frac{\partial^2 F_{\mathrm{succ}}}{\partial F_1^2} = \frac{12(4F_2-1)(1-14F_2+4F_2^2)}{(5-2F_1-2F_2+8F_1F_2)^3},\\
        (F_{\mathrm{succ}})_{F_1F_2} =& \frac{\partial^2 F_{\mathrm{succ}}}{\partial F_1\partial F_2} = \frac{\partial^2 F_{\mathrm{succ}}}{\partial F_2\partial F_1} = (F_{\mathrm{succ}})_{F_2F_1} = \frac{18(11-2F_1-2F_2-16F_1F_2)}{(5-2F_1-2F_2+8F_1F_2)^3},\\
        (F_{\mathrm{succ}})_{F_2F_2} =& \frac{\partial^2 F_{\mathrm{succ}}}{\partial F_2^2} = \frac{12(4F_1-1)(1-14F_1+4F_1^2)}{(5-2F_1-2F_2+8F_1F_2)^3}.
    \end{align}
    Then the Hessian is:
    \begin{equation}
        H_{F_{\mathrm{succ}}} = 
        \begin{pmatrix}
            (F_{\mathrm{succ}})_{F_1F_1} & (F_{\mathrm{succ}})_{F_1F_2} \\
            (F_{\mathrm{succ}})_{F_2F_1} & (F_{\mathrm{succ}})_{F_2F_2}
        \end{pmatrix}.
    \end{equation}
    The Hessian $H_{F_{\mathrm{succ}}}$ is symmetric, and thus it is negative semidefinite if and only if $(F_{\mathrm{succ}})_{F_1F_1}\leq 0$ and $\vert H_{F_{\mathrm{succ}}}\vert >0$.

    We first evaluate the partial derivative with respect to $F_1$ for $(F_{\mathrm{succ}})_{F_1F_1}$ as follows:
    \begin{align}
        (F_{\mathrm{succ}})_{F_1F_1F_1} = \frac{\partial}{\partial F_1}(F_{\mathrm{succ}})_{F_1F_1} = \frac{72(4F_2-1)^2(14F_2-4F_2^2-1)}{(5-2F_1-2F_2+8F_1F_2)^4}.
    \end{align}
    It is obvious that $(F_{\mathrm{succ}})_{F_1F_1F_1}\geq 0$ for $(F_1,F_2)\in[1/2,1]\times[1/2,1]$. Therefore, the maximal value can only be taken on the boundary of $F_1=1$. Then on this boundary we have:
    \begin{align}
        \left.(F_{\mathrm{succ}})_{F_1F_1}\right\vert_{F_1=1} = \frac{4(4F_2-1)(1-14F_2+4F_2^2)}{9(1+2F_2)^3},
    \end{align}
    which is obviously negative for $F_2\in[1/2,1]$. Hence, we have shown that the maximal value of $(F_{\mathrm{succ}})_{F_1F_1}$ on $[1/2,1]\times[1/2,1]$, i.e. $(F_{\mathrm{succ}})_{F_1F_1}<0$.

    Then, the determinant of the Hessian can also be straightforwardly evaluated as:
    \begin{align}
        \vert H_{F_{\mathrm{succ}}}\vert = \frac{36[128F_1^2F_2^2 + 1208F_1F_2 - 217 - 22(F_1+F_2) + 32(F_1^2+F_2^2) - 448(F_1^2F_2+F_1F_2^2)]}{(5-2F_1-2F_2+8F_1F_2)^5}.
    \end{align}
    The denominator is positive on $[1/2,1]\times[1/2,1]$, so we only need to examine the numerator $g_H$ (without the coefficient 36). Again, we take the partial derivative of the numerator with respect to $F_1$:
    \begin{align}
        \frac{\partial}{\partial F_1}g_H =& 256F_1F_2^2 - 448F_2^2 - 896F_1F_2 + 1208F_2 + 64F_1 - 22,
    \end{align}
    whose partial derivative with respect to $F_1$ is:
    \begin{align}
        \frac{\partial^2}{\partial F_1^2}g_H = 64(1 - 14F_2 + 4F_2^2),
    \end{align}
    which is negative for $F_2\in[1/2,1]$. Therefore, $\partial g_H/\partial F_1$ will take minimal value on the boundary of $F_1=1$. Then on this boundary we have:
    \begin{align}
        \left.\frac{\partial}{\partial F_1}g_H\right\vert_{F_1=1} =& 42 + 312F_2 - 192F_2^2,
    \end{align}
    which is positive for $F_2\in[1/2,1]$. Therefore, we have that $\partial g_H/\partial F_1>0$ on $[1/2,1]\times[1/2,1]$. According to the symmetry of $g_H$ as function of $F_1$ and $F_2$, we also have that $\partial g_H/\partial F_2>0$ on $[1/2,1]\times[1/2,1]$. At $(1/2,1/2)$, we can calculate that $g_H(1/2,1/2) = -25<0$, while at $(1,1)$ we have that $g_H(1,1) = 243>0$. Then according to the monotonicity, continuity and symmetry of $g_H$, we are sure that there exist $(F_{\mathrm{th}},F_{\mathrm{th}})$ s.t. $g_H(F_{\mathrm{th}},F_{\mathrm{th}})=0$ and $g_H(F_1,F_2)>0$ for $F_1,F_2\in(F_{\mathrm{th}},1]$. 
    
    Notice that with respect to $F_1$ or $F_2$, $g_H$ is simply a quadratic function. Therefore, we can find $F_{\mathrm{th}}$ by explicitly solving the equation $g_H=0$ for either $F_1$ or $F_2$. The one solution that is relevant in our case is:
    \begin{align}
        F_1 = \frac{11 - 604F_2 + 224F_2^2 + 3\sqrt{5}\sqrt{157 - 2440F_2 + 8592F_2^2 - 5632F_2^3 + 1024F_2^4}}{32(1-14F_2+4F_2^2)} = f(F_2).
    \end{align}        
    Then $F_{\mathrm{th}}$ is just the solution of $F_2 = f(F_2)$ in the interval $[1/2,1]$:
    \begin{align}
        F_{\mathrm{th}} = \frac{7}{4} - \frac{3\times 5^{1/3}}{8}\left[(1-i\sqrt{3})(2-i)^{1/3} + (1+i\sqrt{3})(2+i)^{1/3}\right] \approx 0.54008.
    \end{align}
    
    Combining the negativity of $(F_{\mathrm{succ}})_{F_1F_1}$, and positivity of $\vert H_{F_{\mathrm{succ}}}\vert$, we conclude that $F_{\mathrm{succ}}(F_1,F_2)$ is concave on $[F_{\mathrm{th}},1]\times[F_{\mathrm{th}},1]$. Finally, we comment that it can be shown $f(F_2)$ is also convex on $[1/2,1]$. Therefore, $F_{\mathrm{succ}}(F_1,F_2)$ is concave on the entire convex region within $[1/2,1]\times[1/2,1]$ where $\vert H_{F_{\mathrm{succ}}}\vert\geq 0$, which is larger than $[F_{\mathrm{th}},1]\times[F_{\mathrm{th}},1]$.
\end{proof}

However, we comment that according to the proof the output fidelity is concave for fixed $F_1$ or $F_2$.

\section{Continuous-time memory decoherence}\label{sec:decoherence_model}
In this section, we review continuous-time Pauli error channel model, and examine the impact of buffer memory decoherence on AS performance in more details.

\subsection{Continuous-time Pauli error channel}
Here we review the continuous-time Pauli error channel model~\cite{zang2025entanglement} which is used to prove the monotonicity of DEJMPS output fidelity given two arbitrary BDS initial states undergoing varying duration of memory idling decoherence with any possible pattern of Pauli errors.

We consider continuous-time Pauli channel defined by the channel over a time step $\Delta t$:
\begin{align}\label{eqn:pauli_def_discrete}
    \mathcal{E}_{\Delta t}(\rho) = (1-p)\rho + p(aX\rho X + bY\rho Y + cZ\rho Z)
\end{align}
where $p$ is the probability of having any Pauli error to $\rho$ during time step $\Delta t$, and $a+b+c=1$ correspond to relative strength of three Pauli errors. We can derive the Pauli channel for duration $t$:
\begin{align}
    \tilde{\mathcal{E}}_P(t) =& \frac{1 + e^{-2(\gamma_x+\gamma_y)t} + e^{-2(\gamma_y+\gamma_z)t} + e^{-2(\gamma_x+\gamma_z)t}}{4}\mathrm{Id} + \frac{1 - e^{-2(\gamma_x+\gamma_y)t} + e^{-2(\gamma_y+\gamma_z)t} - e^{-2(\gamma_x+\gamma_z)t}}{4}\mathrm{X}\notag\\
    &+ \frac{1 - e^{-2(\gamma_x+\gamma_y)t} - e^{-2(\gamma_y+\gamma_z)t}+ e^{-2(\gamma_x+\gamma_z)t}}{4}\mathrm{Y} + \frac{1 + e^{-2(\gamma_x+\gamma_y)t} - e^{-2(\gamma_y+\gamma_z)t} - e^{-2(\gamma_x+\gamma_z)t}}{4}\mathrm{Z}\notag\\
    =& p_I(t)\mathrm{Id} + p_X(t)\mathrm{X} + p_Y(t)\mathrm{Y} + p_Z(t)\mathrm{Z},
    \label{eqn:p_i(t)_def}
\end{align}
where capital letters X,Y,Z denotes the quantum channels defined by conjugate of Pauli operators $X,Y,Z$, respectively, and Id denotes the identity channel. In entanglement distribution scenarios, we can assume that both quantum memories on two parties undergo same duration of idling. The above equation is also an implicit definition of functions $p_I(t),p_X(t),p_Y(t),p_Z(t)$, where we have defined error rate $\gamma\Delta t=p$ and $\gamma_x=a\gamma,\gamma_y=b\gamma,\gamma_z=c\gamma$. Then, the combined error channel on an entangled state stored in the two quantum memories is a tensor product of the two local memory idling channels:
\begin{align}
    \tilde{\mathcal{E}}_{P,AB}(t) = \tilde{\mathcal{E}}_{P,A}(t)\otimes\tilde{\mathcal{E}}_{P,B}(t),
\end{align}
where $A,B$ denote the quantum memory on party A, B, respectively. We have the final state of an arbitrary BDS $\rho=\nu_1\Phi^+ + \nu_2\Phi^- + \nu_3\Psi^+ + \nu_4\Psi^-$ undergoing $\tilde{\mathcal{E}}_{P,AB}(t)$ as:
\begin{align}
    \begin{pmatrix}
        \nu_1(t)\\
        \nu_2(t)\\
        \nu_3(t)\\
        \nu_4(t)\\
    \end{pmatrix} =
    \begin{pmatrix}
        C_I(t) & C_Z(t) & C_X(t) & C_Y(t) \\
        C_Z(t) & C_I(t) & C_Y(t) & C_X(t) \\
        C_X(t) & C_Y(t) & C_I(t) & C_Z(t) \\
        C_Y(t) & C_X(t) & C_Z(t) & C_I(t) \\        
    \end{pmatrix}
    \begin{pmatrix}
        \nu_1\\
        \nu_2\\
        \nu_3\\
        \nu_4\\
    \end{pmatrix},
\end{align}
where 
\begin{align}
    C_I(t) =& p_I^{(A)}(t)p_I^{(B)}(t) + p_X^{(A)}(t)p_X^{(B)}(t) + p_Y^{(A)}(t)p_Y^{(B)}(t) + p_Z^{(A)}(t)p_Z^{(B)}(t),\\
    C_X(t) =& p_I^{(A)}(t)p_X^{(B)}(t) + p_X^{(A)}(t)p_I^{(B)}(t) + p_Y^{(A)}(t)p_Z^{(B)}(t) + p_Z^{(A)}(t)p_Y^{(B)}(t),\\
    C_Y(t) =& p_I^{(A)}(t)p_Y^{(B)}(t) + p_Y^{(A)}(t)p_I^{(B)}(t) + p_X^{(A)}(t)p_Z^{(B)}(t) + p_Z^{(A)}(t)p_X^{(B)}(t),\\
    C_Z(t) =& p_I^{(A)}(t)p_Z^{(B)}(t) + p_Z^{(A)}(t)p_I^{(B)}(t) + p_X^{(A)}(t)p_Y^{(B)}(t) + p_Y^{(A)}(t)p_X^{(B)}(t).    
\end{align}

We consider two initial BDS:
\begin{align}
    & \rho_1(0)=
    \begin{pmatrix}
        F_1\\
        a_1(1-F_1)\\
        b_1(1-F_1)\\
        (1-a_1-b_1)(1-F_1)
    \end{pmatrix},\\
    & \rho_2(0)=
    \begin{pmatrix}
        F_2\\
        a_2(1-F_2)\\
        b_2(1-F_2)\\
        (1-a_2-b_2)(1-F_2)
    \end{pmatrix},
\end{align}
where $F_1,F_2\in[1/2,1]$, $a_{1(2)}\in[0,1]$, and $b_{1(2)}\in[0,1-a_{1(2)}]$. We assume that quantum memories on party A all undergo an identical continuous-time Pauli channel characterized by $(\gamma_{x,A},\gamma_{y,A},\gamma_{z,A})$, and quantum memories on party B all undergo an identical continuous-time Pauli channel characterized by $(\gamma_{x,B},\gamma_{y,B},\gamma_{z,B})$. We further consider that $\rho_1(0)$ has undergone $t_1$ memory idling which results in $\rho_1(t_1)$, while $\rho_2(0)$ has undergone $t_2$ memory idling which results in $\rho_2(t_2)$. 

\subsection{Output with Werner input states under depolarizing channel}
We have explicit formulae for the successful output fidelity and success probability of the DEJMPS protocol, given two Werner input states $\rho_1(j)$ and $\rho_2(k)$, where $\rho_i(s), i=1,2$ denotes the final state after the initial entangled state $\rho_i(0)$ undergoes $sT$ duration of memory decoherence, recalling that $T$ is the cycle time of entanglement distribution
\begin{align}
    p_{\mathrm{succ}}\left(\rho_1(j)\otimes\rho_2(k)\right) =& \frac{9+e^{-(j+k)r}(4F_1-1)(4F_2-1)}{18},\\
    F_{\mathrm{succ}}\left(\rho_1(j)\otimes\rho_2(k)\right) =& \frac{9 + 3e^{-jr}(4F_1-1) + 3e^{-kr}(4F_2-1) + 5e^{-(j+k)r}(4F_1-1)(4F_2-1)}{72p_{\mathrm{succ}}\left(\rho_1(j)\otimes\rho_2(k)\right)},
\end{align}
where we have defined $r=2\kappa T$.

Then we can express the success probability and successful output fidelity under the assumed scenario as:
\begin{align}
    p_{\mathrm{succ}}\left(\rho_{\mathrm{eff},n=2}^{(m)}\right) =& \frac{\sum_{i=1,2}\sum_{0\leq j<k<m}p_{\mathrm{succ}}(\rho_i(j)\otimes\rho_i(k)) + \sum_{j,k=0}^{m-1}p_{\mathrm{succ}}(\rho_1(j)\otimes\rho_2(k))}{\binom{2m}{2}},\\
    F_{\mathrm{succ}}\left(\rho_{\mathrm{eff},n=2}^{(m)}\right) =& \frac{\sum_{i=1,2}\sum_{0\leq j<k<m}p_{\mathrm{succ}}(\rho_i(j)\otimes\rho_i(k))F_{\mathrm{succ}}(\rho_i(j)\otimes\rho_i(k)) + \sum_{j,k=0}^{m-1}p_{\mathrm{succ}}(\rho_1(j)\otimes\rho_2(k))F_{\mathrm{succ}}(\rho_1(j)\otimes\rho_2(k))}{\binom{2m}{2}p_{\mathrm{succ}}\left(\rho_{\mathrm{eff},n=2}^{(m)}\right)}.
\end{align}

\subsection{Accumulating and shuffling under memory decoherence}
Then we immediately have the following intuitive statement about success probability of DEJMPS protocol given $\rho_1(t_1)$ and $\rho_2(t_2)$ as two input states:
\begin{proposition}
    $p_{\mathrm{succ}}\left(\rho_1(t_1)\otimes\rho_2(t_2)\right)$ decreases monotonically as $t_1$ or $t_2$ increases, for $F_1,F_2\in[1/2,1]$.
\end{proposition}
\begin{proof}
    We can directly write out the analytical expression for $p_{\mathrm{succ}}\left(\rho_1(t_1)\otimes\rho_2(t_2)\right)$:
    \begin{align}
        p_{\mathrm{succ}}\left(\rho_1(t_1)\otimes\rho_2(t_2)\right) = \left[1 + e^{-2(\gamma_{x,A}+\gamma_{y,A}+\gamma_{x,B}+\gamma_{y,B})(t_1+t_2)}(2F_1-2a_1F_1+2a_1-1)(2F_2-2a_2F_2+2a_2-1)\right]/2,
    \end{align}
    which obviously decreases monotonically as $t_1$ or $t_2$ increases, given positive $\gamma_{x,A},\gamma_{y,A},\gamma_{x,B},\gamma_{y,B}$, $F_1,F_2\in[1/2,1]$, and $a_1,a_2\in[0,1]$.
\end{proof}

For successful output fidelity, the analytical expression becomes much more complicated, and symbolic analysis becomes hard for its monotonicity. Therefore, we prove a more restrictive version of its monotonicity with respect to $t_1$ and $t_2$, for which we first prove the following Lemma regarding the monotonicity of DEJMPS successful output fidelity given two arbitrary BDS input states:
\begin{lemma}\label{thm:DEJMPS_BDS_fid_monotonicity}
    For two BDS $\rho_1=(F_1,a_1,b_1,c_1)$ and $\rho_2=(F_2,a_2,b_2,c_2)$ as input states to the DEJMPS protocol, the successful output fidelity $F_{\mathrm{succ}}\left(\rho_1\otimes\rho_2\right)$ monotonically decreases if $F_{1(2)}\in[2-\sqrt{2},1]$ decreases, or $a_{1(2)}, b_{1(2)}, c_{1(2)}$ increase.
\end{lemma}
\begin{proof}
    The analytical expression of $F_{\mathrm{succ}}\left(\rho_1\otimes\rho_2\right)$ is:
    \begin{align}
        F_{\mathrm{succ}}\left(\rho_1\otimes\rho_2\right) = \frac{a_1a_2 + F_1F_2}{(b_1+c_1)(b_2+c_2) + (a_1+F_1)(a_2+F_2)}.
    \end{align}
    By the symmetry between $\rho_1$ and $\rho_2$, we only need to examine the partial derivatives with respect to $F_1,a_1,b_1,c_1$:
    \begin{align}
        \frac{\partial}{\partial F_1}F_{\mathrm{succ}}\left(\rho_1\otimes\rho_2\right) =& \frac{(b_1+c_1)(b_2+c_2)F_2 + a_1(F_2^2-a_2^2)}{[(b_1+c_1)(b_2+c_2) + (a_1+F_1)(a_2+F_2)]^2}, \\
        \frac{\partial}{\partial a_1}F_{\mathrm{succ}}\left(\rho_1\otimes\rho_2\right) =& \frac{a_2(b_1+c_1)(b_2+c_2) - F_1(F_2^2-a_2^2)}{[(b_1+c_1)(b_2+c_2) + (a_1+F_1)(a_2+F_2)]^2}, \\
        \frac{\partial}{\partial b_1}F_{\mathrm{succ}}\left(\rho_1\otimes\rho_2\right) =& \frac{\partial}{\partial c_1}F_{\mathrm{succ}}\left(\rho_1\otimes\rho_2\right) = -\frac{(b_2+c_2)(a_1a_2 + F_1F_2)}{[(b_1+c_1)(b_2+c_2) + (a_1+F_1)(a_2+F_2)]^2},
    \end{align}
    where the first partial derivative is obviously positive and last line is obviously negative. For the partial derivative with respect to $a_1$, we analyze the numerator as:
    \begin{align}
        & a_2(b_1+c_1)(b_2+c_2) - F_1(F_2^2-a_2^2)\notag\\
        \leq& a_2(1-F_1)(1-F_2) - F_1(F_2^2-a_2^2) \notag\\
        \leq& F_1[a_2^2 + a_2(1-F_2) - F_2^2] \notag\\
        \leq& F_1[2(1-F_2)^2 - F_2^2],
    \end{align}
    which is guaranteed to be non-positive if $F_2\in[2-\sqrt{2},1]$.
\end{proof}

It is noteworthy that although the fidelity always decreases as idling time increases, this is not guaranteed for the other three error components, dependent on the form of initial BDS and memory decoherence channel. Therefore, when we add constraints on the decoherence dynamics of BDS we have the following Proposition as a Corollary of Lemma~\ref{thm:DEJMPS_BDS_fid_monotonicity}:
\begin{proposition}
    When the memory decoherence channels monotonically increases all three error components of the BDS, $F_{\mathrm{succ}}\left(\rho_1(t_1)\otimes\rho_2(t_2)\right)$ decreases monotonically as $t_1$ or $t_2$ increases, for $F_1,F_2\in[2-\sqrt{2},1]$.
\end{proposition}
We comment that the monotonic increase of all three error components of BDS under continuous-time Pauli channel is common, such as when two quantum memories both undergo depolarizing channels:
\begin{lemma}\label{thm:prob_fid_mono_dec_depolarizing}
    $p_{\mathrm{succ}}\left(\rho_1(j)\otimes\rho_2(k)\right)$ and $F_{\mathrm{succ}}\left(\rho_1(j)\otimes\rho_2(k)\right)$ decrease monotonically as $j$ or $k$ increases, $\forall F_1,F_2\in[1/2,1]$, where $\rho_1(j),\rho_2(k)$ are the same as used in the main text.
\end{lemma}
\begin{proof}
    According to the symmetry of $j$ and $k$ in these formulae, we can focus on the monotonicity with respect to $j$ only. We take the partial derivative of $p_{\mathrm{succ}}\left(\rho_1(j)\otimes\rho_2(k)\right)$ with respect to $j$:
    \begin{align}
        \frac{\partial}{\partial j}p_{\mathrm{succ}}\left(\rho_1(j)\otimes\rho_2(k)\right) = -e^{-(j+k)r}(4F_1-1)(4F_2-1)r/18 <0,
    \end{align}
    for $F_1,F_2\in[1/2,1]$. Then we take the partial derivative of $F_{\mathrm{succ}}\left(\rho_1(j)\otimes\rho_2(k)\right)$ with respect to $j$:
    \begin{align}
        \frac{\partial}{\partial j}F_{\mathrm{succ}}\left(\rho_1(j)\otimes\rho_2(k)\right) = \frac{3e^{jr}(4F_1-1)\left[(4F_2-1)^2-12e^{kr}(4F_2-1)-9e^{2kr}\right]r}{4\left[9e^{(j+k)r}+(4F_1-1)(4F_2-1)\right]^2}.
    \end{align}
    It is negative for $F_1,F_2\in[1/2,1]$ because the bracket term in the numerator equals $36\left[e^{kr}-(4F_2-1)/6\right]^2-45e^{2kr}$, which is negative for $r>0$ and $F_2\in[1/2,1]$.
\end{proof}

\subsection{Discussion on the number of entanglement distribution rounds under memory decoherence}\label{sec:proof_m=2>m=3}
We have the following result.
\begin{proposition}\label{thm:m=2>m=3_with_decoherence}
    Under the above assumption of initial Werner states with arbitrary fidelities $F_1,F_2\in[1/2,1]$ and identical quantum memory depolarizing channels whose depolarizing rate is characterized by $r$, we have $p_{\mathrm{succ}}\left(\rho_{\mathrm{eff},n=2}^{(2)}\right)\geq p_{\mathrm{succ}}\left(\rho_{\mathrm{eff},n=2}^{(3)}\right)$ for $r\geq 0.036$, and $\tilde{F}_{\mathrm{succ}}\left(\rho_{\mathrm{eff},n=2}^{(2)}\right)\geq \tilde{F}_{\mathrm{succ}}\left(\rho_{\mathrm{eff},n=2}^{(3)}\right)$ for $r\geq 0.027$.
\end{proposition}
We separate the proof for success probability and success-weighted fidelity, respectively.
\begin{proof}
    \textbf{\textit{Success probability.}} We directly take the difference:
    \begin{align}
        p_{\mathrm{succ}}\left(\rho_{\mathrm{eff},n=2}^{(2)}\right) - p_{\mathrm{succ}}\left(\rho_{\mathrm{eff},n=2}^{(3)}\right) = \frac{e^{-4r}}{540}g_p(F_1,F_2,r),
    \end{align}
    where
    \begin{align}
        g_p(F_1,F_2,r) =& 3e^{4r}(4F_1-1)(4F_2-1) + 12e^{3r}(2F_1+2F_2-1)^2 \notag\\
        &- e^{2r}(5-20F_1-20F_2+32F_1^2+16F_1F_2+32F_2^2)\notag\\
        &- 8e^r(2F_1+2F_2-1)^2 - (2-8F_1-8F_2+32F_1F_2).
    \end{align}
    
    Then we perform a coordinate transformation from horizontal and vertical coordinates to diagonal and off-diagonal coordinates, defined by $F_+=(F_1+F_2)/\sqrt{2}$ and $F_-=(F_2-F_1)/\sqrt{2}$, or equivalently $F_1=(F_+-F_-)/\sqrt{2}$ and $F_2=(F_++F_-)/\sqrt{2}$. We then take the partial derivative with respect to $F_-$:
    \begin{align}
        &\frac{\partial}{\partial F_-}g_p\left(\frac{F_+-F_-}{\sqrt{2}},\frac{F_++F_-}{\sqrt{2}},r\right) = 32(1-3e^{3r}\cosh r)F_-,
    \end{align}
    where the factor multiplying $F_-$ is negative for valid $r\geq 0$. This suggests that the minimum value of $g_p(F_1,F_2,r)$ be on the boundaries of $[1/2,1]\times[1/2,1]$, which means that we can guarantee its positivity on this 2-d region if we can show that its values on the boundaries are all non-negative. Moreover, given the symmetry of $F_1$ and $F_2$ for $g_p(F_1,F_2,r)$, we only need to consider the $F_2=1/2$ and $F_1=1$ boundaries.

    For $F_2=1/2$, we can examine the partial derivative with respect to $F_1$:
    \begin{align}
        \frac{\partial}{\partial F_1}g_p(F_1,1/2,r) = 32e^{2r}F_1(\cosh r + 5\sinh r - 2) + 4e^{2r}(3 + \cosh 2r + 5\sinh 2r),
    \end{align}
    which takes zero value at $F_1 = \frac{-3 - \cosh 2r - 5\sinh 2r}{8(\cosh r + 5\sinh r - 2)}$, whose denominator monotonically increases for $r\geq 0$ and becomes zero at $r=\log[(1+\sqrt{7})/3]\approx 0.195$. Thus the zero point is negative for $r> 0.195$ and positive for $0\leq r< 0.195$. Further partial derivative gives:
    \begin{align}
        \frac{\partial^2}{\partial F_1^2}g_p(F_1,1/2,r) = 32e^{2r}(\cosh r + 5\sinh r - 2),
    \end{align}
    which means that when the zero point of first order derivative is positive, $g_p(F_1,1/2,r)$ is also concave. According to the above properties, $g_p(F_1,1/2,r)$ is either concave or monotonically convex on $F_1\in[1/2,1]$. Therefore, on this boundary the minimum value can only be obtained on either $F_1=1/2$ or $F_1=1$.

    For $F_1=1$, we can examine the partial derivative with respect to $F_2$:
    \begin{align}
        \frac{\partial}{\partial F_2}g_p(1,F_2,r) = 4(9e^{4r}+12e^{3r}+e^{2r}-8e^r-6) + 32(3e^{3r} - 2e^{2r} - 2e^r)F_2.
    \end{align}
    We can then take partial derivative with respect to $r$ of the above result and obtain:
    \begin{align}
        \frac{\partial^2}{\partial r\partial F_2}g_p(1,F_2,r) = 8(18e^{4r}+18e^{3r}+e^{2r}-4e^r) + 32(9e^{3r} - 4e^{2r} - 2e^r)F_2,
    \end{align}
    which is obviously positive for $r\geq 0$ and $F_2\in[1/2,1]$. Therefore, we could first examine $\partial g_p(1,F_2,r)/\partial F_2$ at $r=0$:
    \begin{align}
        \left.\frac{\partial}{\partial F_2}g_p(1,F_2,r)\right\vert_{r=0} = 32 - 32F_2\geq 0,
    \end{align}
    from which we know that $g_p(1,F_2,r)$ monotonically increases as $F_2\in[1/2,1]$ increases for $r\geq 0$.

    According to the above discussions of the two boundaries, we now examine $g_p(F_1,F_2,r)$ at $(1/2,1/2)$ and $(1,1/2)$:
    \begin{align}
        & g_p(1/2,1/2,r) = (e^r-1)(2+10e^r+15e^{2r}+3e^{3r}),\\
        & g_p(1,1/2,r) = 9e^{4r}+48e^{3r}-23e^{2r}-32e^r-6,
    \end{align}
    where it is obvious that $g_p(1/2,1/2,r)\geq 0$. We are thus sure that $p_{\mathrm{succ}}\left(\rho_{\mathrm{eff},n=2}^{(2)}\right)\geq p_{\mathrm{succ}}\left(\rho_{\mathrm{eff},n=2}^{(3)}\right)$ is satisfied for arbitrary $F_1,F_2\in[1/2,1]$ as long as $g_p(1,1/2,r)\geq 0$. We can see that $g_p(1,1/2,r)$ increases monotonically as $r\geq 0$ increases, and we can also obtain the numerical value of $r$ when $g_p(1,1/2,r)=0$ as $r\approx 0.036$.
\end{proof}
\begin{proof}
    \textbf{\textit{Success-weighted fidelity.}} The approach will be the same to the discussion on success probability. We directly take the difference:
    \begin{align}
        \tilde{F}_{\mathrm{succ}}\left(\rho_{\mathrm{eff},n=2}^{(2)}\right) - \tilde{F}_{\mathrm{succ}}\left(\rho_{\mathrm{eff},n=2}^{(3)}\right) = \frac{e^{-4r}}{432}g_{\tilde{F}}(F_1,F_2,r),
    \end{align}
    where 
    \begin{align}
        g_{\tilde{F}}(F_1,F_2,r) =& e^{4r}(48F_1F_2-3) + 6e^{3r}(2F_1+2F_2-1)(4F_1+4F_2-1) \notag\\
        & + e^{2r}(7-4F_1-4F_2-32F_1^2-16F_1F_2-32F_2^2) \notag\\
        & - 8e^r(2F_1+2F_2-1)^2 - (2-8F_1-8F_2+32F_1F_2).
    \end{align}

    Then we perform the same coordinate transform and take the partial derivative with respect to $F_-$:
    \begin{align}
        \frac{\partial}{\partial F_-}g_{\tilde{F}}\left(\frac{F_+-F_-}{\sqrt{2}},\frac{F_++F_-}{\sqrt{2}},r\right) = 32(1-3e^{3r}\cosh r)F_-,
    \end{align}
    which is curiously identical to the one for $g_p$. Therefore, we can again safely examine the $F_2=1/2$ and $F_1=1$ boundaries only.

    For $F_2=1/2$ we have:
    \begin{align}
        \frac{\partial}{\partial F_1}g_{\tilde{F}}(F_1,1/2,r) = 32(3e^{3r} - 2e^{2r} - 2e^r)F_1 + 4(6e^{4r}+3e^{3r}-3e^{2r}-2),
    \end{align}
    which takes zero value at $F_1 = \frac{2-6e^{4r}-3e^{3r}+3e^{2r}}{8(3e^{3r} - 2e^{2r} - 2e^r)}$, whose denominator monotonically increases for $r\geq 0$ and becomes zero again at $r=\log[(1+\sqrt{7})/3]\approx 0.195$. Thus the zero point is negative for $r> 0.195$ and positive for $0\leq r< 0.195$. Further partial derivative gives:
    \begin{align}
        \frac{\partial^2}{\partial F_1^2}g_{\tilde{F}}(F_1,1/2,r) = 32(3e^{3r} - 2e^{2r} - 2e^r),
    \end{align}
    which means that when the zero point of first order derivative is positive, $g_{\tilde{F}}(F_1,1/2,r)$ is also concave. Thus, $g_{\tilde{F}}(F_1,1/2,r)$ is either concave or monotonically convex on $F_1\in[1/2,1]$. Therefore, on this boundary the minimum value can only be obtained on either $F_1=1/2$ or $F_1=1$.

    For $F_1=1$ we have:
    \begin{align}
        \frac{\partial}{\partial F_2}g_{\tilde{F}}(1,F_2,r) = 4(12e^{4r}+15e^{3r}-5e^{2r}-8e^r-6) + 32(3e^{3r} - 2e^{2r} - 2e^r)F_2,
    \end{align}
    We can then take partial derivative with respect to $r$ of the above result and obtain:
    \begin{align}
        \frac{\partial^2}{\partial r\partial F_2}g_{\tilde{F}}(1,F_2,r) = 4(48e^{4r}+45e^{3r}-10e^{2r}-8e^r) + 32(9e^{3r} - 4e^{2r} - 2e^r)F_2,
    \end{align}
    which is obviously positive for $r\geq 0$ and $F_2\in[1/2,1]$. Therefore, we could first examine $\partial g_{\tilde{F}}(1,F_2,r)/\partial F_2$ at $r=0$:
    \begin{align}
        \left.\frac{\partial}{\partial F_2}g_{\tilde{F}}(1,F_2,r)\right\vert_{r=0} = 32 - 32F_2\geq 0,
    \end{align}
    from which we know that $g_{\tilde{F}}(1,F_2,r)$ monotonically increases as $F_2\in[1/2,1]$ increases for $r\geq 0$.

    Then same as in the first part, we only need to examine $g_{\tilde{F}}(F_1,F_2,r)$ at $(1/2,1/2)$ and $(1,1/2)$:
    \begin{align}
        & g_{\tilde{F}}(1/2,1/2,r) = (e^r-1)(2+10e^r+27e^{2r}+9e^{3r}),\\
        & g_{\tilde{F}}(1,1/2,r) = 21e^{4r}+60e^{3r}-47e^{2r}-32e^r-6,
    \end{align}
    where $g_{\tilde{F}}(1/2,1/2,r)\geq 0$. We are thus sure that $\tilde{F}_{\mathrm{succ}}\left(\rho_{\mathrm{eff},n=2}^{(2)}\right)\geq \tilde{F}_{\mathrm{succ}}\left(\rho_{\mathrm{eff},n=2}^{(3)}\right)$ is satisfied for arbitrary $F_1,F_2\in[1/2,1]$ as long as $g_{\tilde{F}}(1,1/2,r)\geq 0$. We can see that $g_{\tilde{F}}(1,1/2,r)$ increases monotonically as $r\geq 0$ increases, and we can also obtain the numerical value of $r$ when $g_{\tilde{F}}(1,1/2,r)=0$ as $r\approx 0.027$.
\end{proof}

\begin{figure}[t]
    \centering
    \includegraphics[width=0.95\linewidth]{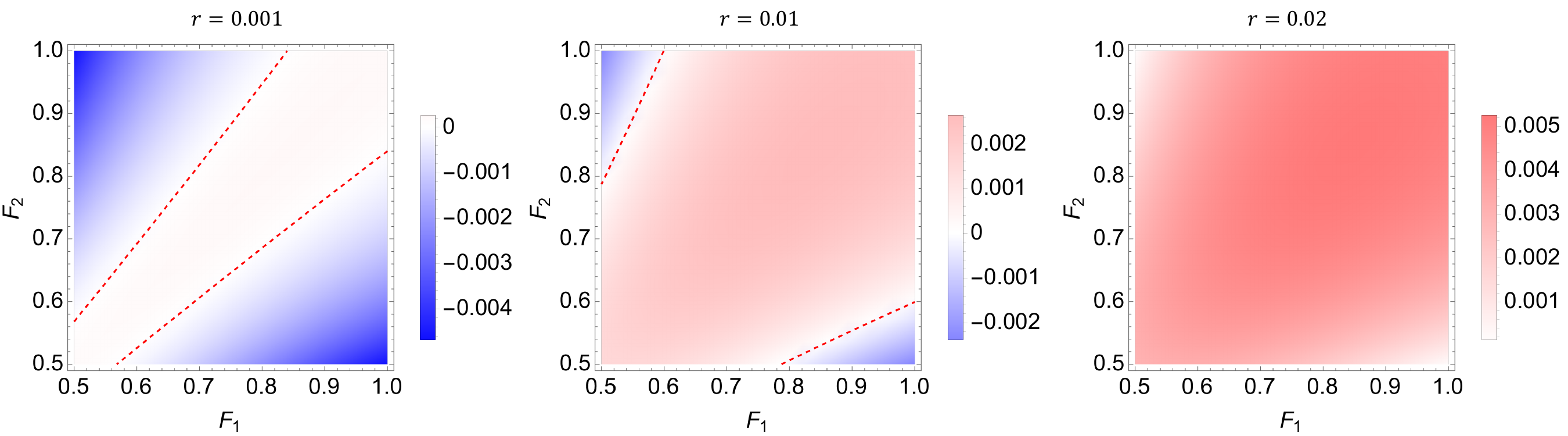}
    \caption{Difference between $F_{\mathrm{succ}}\left(\rho_{\mathrm{eff},n=2}^{(2)}\right)$ and $F_{\mathrm{succ}}\left(\rho_{\mathrm{eff},n=2}^{(3)}\right)$ for $r=0.001,0.01,0.02$. White color is fixed to represent zero value, while blue color denotes negative value and red color corresponds to positive value. The red dashed lines illustrate the boundaries where the difference is exactly zero.}
    \label{fig:m=2vsm=3_fid}
\end{figure}

Other figures of merit such as successful output fidelity also demonstrate similar behaviors under varying $r$, while having different threshold values of $r$. Explicitly, we visualize the difference $F_{\mathrm{succ}}\left(\rho_{\mathrm{eff},n=2}^{(2)}\right)-F_{\mathrm{succ}}\left(\rho_{\mathrm{eff},n=2}^{(3)}\right)$ under different values of $r$ in Fig.~\ref{fig:m=2vsm=3_fid}. We can indeed observe that when input states are similar to each other, i.e. $F_1\approx F_2$, the performance of $m=2$ AS is always better than that of $m=3$ AS even at small $r$, while when input states differ significantly, $m=3$ AS outperforms its $m=2$ counterpart when decoherence rate is small. Additionally, the advantage of $m=2$ AS is demonstrated in larger fidelity combination region while the decoherence rate increases, and finally when the decoherence rate surpasses certain threshold that depends on both figure of merit and error model (given depolarizing channel and initial Werner state in this section, the threshold for guaranteed fidelity advantage is $r\approx 0.0194$), $m=2$ AS is guaranteed to be better than $m=3$ AS for arbitrary combination of $\rho_1$ and $\rho_2$.

\subsection{Proof of main text Proposition~\ref{thm:m=2<NoAS_with_decoherence}}\label{sec:proof_m=2<NoAS}
\begin{proposition}[Restatement of main text Proposition~\ref{thm:m=2<NoAS_with_decoherence}]
    For raw Werner states with arbitrary fidelities $F_1,F_2\in[1/2,1]$ and memory depolarizing channels whose depolarizing rate is characterized by $r=2\kappa T$, we have $p_{\mathrm{succ}}\left(\rho_{\mathrm{eff},n=2}^{(2)}\right)\leq p_{\mathrm{succ}}\left(F_1,F_2\right)$ for $r\geq 0.206$, $F_{\mathrm{succ}}\left(\rho_{\mathrm{eff},n=2}^{(2)}\right)\leq F_{\mathrm{succ}}\left(F_1,F_2\right)$ for $r\geq 0.106$, and $\tilde{F}_{\mathrm{succ}}\left(\rho_{\mathrm{eff},n=2}^{(2)}\right)\leq \tilde{F}_{\mathrm{succ}}\left(F_1,F_2\right)$ for $r\geq 0.150$.
\end{proposition}

We separate the proof for success probability, fidelity, and success-weighted fidelity, respectively.
\begin{proof}
    \textbf{\textit{Success probability.}} We consider the difference:
    \begin{align}
        p_{\mathrm{succ}}\left(\rho_{\mathrm{eff},n=2}^{(2)}\right) - p_{\mathrm{succ}}\left(F_1,F_2\right) = \frac{e^{-2r}}{108}h_p(F_1,F_2,r),
    \end{align}
    where 
    \begin{align}
        h_p(F_1,F_2,r) = (4F_1-1)(4F_2-1) + 4e^r(2F_1+2F_2-1)^2 - 5e^{2r}(4F_1-1)(4F_2-1).
    \end{align}

    Then, as in the proof of Proposition~\ref{thm:m=2>m=3_with_decoherence}, we use $F_1=(F_+-F_-)/\sqrt{2}$ and $F_2=(F_++F_-)/\sqrt{2}$, and take the partial derivative with respect to $F_-$:
    \begin{align}
        \frac{\partial}{\partial F_-}h_p\left(\frac{F_+-F_-}{\sqrt{2}},\frac{F_++F_-}{\sqrt{2}},r\right) = 16(5e^{2r}-1)F_-,
    \end{align}
    which implies that the maximum value of $h_p(F_1,F_2,r)$ is on the boundary of $[1/2,1]\times[1/2,1]$. Thus we can guarantee its negativity on this 2-d region if we can show that its values on the boundaries are all non-positive. Moreover, given the symmetry of $F_1$ and $F_2$ for $h_p(F_1,F_2,r)$, we only need to consider the $F_2=1/2$ and $F_1=1$ boundaries.

    For $F_2=1/2$, we can examine the second order partial derivative with respect to $F_1$:
    \begin{align}
        \frac{\partial^2}{\partial F_1^2}h_p(F_1,1/2,r) = 32e^r > 0,
    \end{align}
    which means that $h_p(F_1,1/2,r)$ is convex for $r\geq 0$ on $F_1\in[1/2,1]$, so the maximum value must be at either $F_1=1/2$ or $F_1=1$. For $F_1=1$, we can also examine the second order partial derivative with respect to $F_2$:
    \begin{align}
        \frac{\partial^2}{\partial F_2^2}h_p(1,F_2,r) = 32e^r > 0,
    \end{align}
    which again suggests convexity. As a result, the maximal value of $h_p(F_1,F_2,r)$ on $(F_1,F_2)\in[1/2,1]\times[1/2,1]$ must be at $(1/2,1/2)$, or $(1,1)$, or $(1,1/2)$. It is obvious that the value at $(1/2,1/2)$ and $(1,1)$ must be non-positive because they correspond to cases where input state are already identical, we only need to examine:
    \begin{align}
        h_p(1,1/2,r) = 3 + 16e^r - 15e^{2r},
    \end{align}
    which monotonically decreases as $r$ increases and becomes zero at $r\approx 0.206$.
\end{proof}
\begin{proof}
    \textbf{\textit{Fidelity.}} We consider the difference:
    \begin{align}
        F_{\mathrm{succ}}\left(\rho_{\mathrm{eff},n=2}^{(2)}\right) - F_{\mathrm{succ}}\left(F_1,F_2\right) = e^{-r}\frac{h_F(F_1,F_2,r)}{d_F(F_1,F_2,r)},
    \end{align}
    where the numerator is 
    \begin{align}
        h_F(F_1,F_2,r) =& 3e^{2r}(64F_1^2F_2 + 64F_1F_2^2 - 16F_1^2 - 544F_1F_2 - 16F_2^2 + 78F_1 + 78F_2 - 5) \notag\\
        &+ 6e^r(2F_1+2F_2-1)(1 + 26F_1 + 26F_2 - 8F_1^2 + 8F_1F_2 - 8F_2^2)\notag\\
        &+ 3(4F_1-1)(4F_2-1)(7-2F_1-2F_2) ,
    \end{align}
    while the denominator is
    \begin{align}
        d_F(F_1,F_2,r) = 8(5-2F_1-2F_2+8F_1F_2)\left[2(2F_1+2F_2-1)^2 + 4(7-F_1-F_2+4F_1F_2)\cosh r + 27\sinh r\right] >0,
    \end{align}
    for $F_1,F_2\in[1/2,1]$ and $r\geq 0$. Therefore, the positivity of the difference is only determined by $h_F(F_1,F_2,r)$. We then perform the same coordinate transformation as above, and take the partial derivative with respect to $F_-$:
    \begin{align}
        \frac{\partial}{\partial F_-}h_F\left(\frac{F_+-F_-}{\sqrt{2}},\frac{F_++F_-}{\sqrt{2}},r\right) = 48\left[(32-4\sqrt{2}F_+)e^{2r} + (3-6\sqrt{2}F_+)e^r + (2\sqrt{2}F_+-7)\right]F_-,
    \end{align}
    where the bracket term can be shown to be monotonically increasing as $r\geq 0$ increases for $F_1,F_2\in[1/2,1]$. Then we examine the bracket term at $r=0$: $28-8\sqrt{2}F_+> 0$ for $F_1,F_2\in[1/2,1]$. Therefore, the bracket term is always positive in the parameter space of our interest, which means that similar to the discussion on success probability, we only need to focus on the maximum value on the $F_2=1/2$ and $F_1=1$ boundaries.

    For $F_2=1/2$, we examine the second order partial derivative with respect to $F_1$:
    \begin{align}
        \frac{\partial^2}{\partial F_1^2}h_F(F_1,1/2,r) = 48e^r\left[\cosh r + 3(5-4F_1+\sinh r)\right] > 0,
    \end{align}
    for $F_1\in[1/2,1]$ and $r\geq 0$, which suggests the convexity of $h_F(F_1,1/2,r)$ for $r\geq 0$ on $F_1\in[1/2,1]$. For $F_1=1$, we examine the second order partial derivative with respect to $F_2$:
    \begin{align}
        \frac{\partial^2}{\partial F_2^2}h_F(1,F_2,r) = 144e^r(5 - 4F_2 + \cosh r + 3\sinh r) > 0,
    \end{align}
    for $F_2\in[1/2,1]$ and $r\geq 0$, which suggests the convexity of $h_F(1,F_2,r)$ for $r\geq 0$ on $F_2\in[1/2,1]$. According to same argument as for success probability, we only need to consider:
    \begin{align}
        h_F(1,1/2,r) = 36 + 408e^r - 396e^{2r},
    \end{align}
    which monotonically decreases as $r$ increases and becomes zero at $r\approx 0.106$.
\end{proof}
\begin{proof}
    \textbf{\textit{Success-weighted fidelity.}} We consider the difference:
    \begin{align}
        \tilde{F}_{\mathrm{succ}}\left(\rho_{\mathrm{eff},n=2}^{(2)}\right) - \tilde{F}_{\mathrm{succ}}\left(F_1,F_2\right) = \frac{e^{-2r}}{432}h_{\tilde{F}}(F_1,F_2,r),
    \end{align}
    where 
    \begin{align}
        h_{\tilde{F}}(F_1,F_2,r) =& e^{2r}(64F_1 + 64F_2 - 400F_1F_2 - 7) \notag\\
        &+ 2e^r(2F_1 + 2F_2 - 1)(20F_1 + 20F_2 - 1) + 5(4F_1-1)(4F_2-1).
    \end{align}
    Then we perform the same coordinate transformation as above, and take the partial derivative with respect to $F_-$:
    \begin{align}
        \frac{\partial}{\partial F_-}h_{\tilde{F}}\left(\frac{F_+-F_-}{\sqrt{2}},\frac{F_++F_-}{\sqrt{2}},r\right) = 80(5e^{2r}-1)F_-, 
    \end{align}
    which allows us to focus on the maximum value on the $F_2=1/2$ and $F_1=1$ boundaries.

    For $F_2=1/2$, we examine the second order partial derivative with respect to $F_1$:
    \begin{align}
        \frac{\partial^2}{\partial F_1^2}h_{\tilde{F}}(F_1,1/2,r) = 160e^r > 0,
    \end{align}
    which guarantees the convexity of $h_{\tilde{F}}(F_1,1/2,r)$. For $F_1=1$, we examine the second order partial derivative with respect to $F_2$:
    \begin{align}
        \frac{\partial^2}{\partial F_2^2}h_{\tilde{F}}(1,F_2,r) = 160e^r > 0,
    \end{align}
    which guarantees the convexity of $h_{\tilde{F}}(1,F_2,r)$. Again, we only need to consider:
    \begin{align}
        h_{\tilde{F}}(1,1/2,r) = 15 + 116e^r - 111e^{2r},
    \end{align}
    which monotonically decreases as $r$ increases and becomes zero at $r\approx 0.150$.
\end{proof}

\section{Comment on sources which distribute random states} \label{sec:source_ensemble}
We previously assumed that each entanglement source will produce a fixed density matrix each time. We now allow the sources to produce random states. This models the case in which source $i$ does not produce the same density matrix in every round, but instead produces a density matrix drawn from some probability distribution.

Consider $n$ sources and $m$ rounds of entanglement distribution. Let $R_i^{(t)}$ be the random density matrix distributed by source $i\in\{1,\dots,n\}$ in round $t\in\{1,\dots,m\}$. Assume $R_i^{(t)}$ takes values in a set of density matrices $\Omega_i^{(t)}$, with probability measure $\mu_i^{(t)}$. Then we can express the expected state as
\begin{align}
    \bar{\rho}_i^{(t)} = \E\left[R_i^{(t)}\right] = \int_{\Omega_i^{(t)}} \rho\,d\mu_i^{(t)}(\rho).
\end{align}
Since the set of density matrices is convex, $\bar{\rho}_i^{(t)}$ is again a valid density matrix.

For a fixed realization $r=\{r_i^{(t)}:1\leq i\leq n,\ 1\leq t\leq m\}$, where $r_i^{(t)}$ is a realization of the random density matrix $R_i^{(t)}$, the $mn$ distributed states are randomly shuffled and divided into $m$ packages of $n$ states each. Let $\rho_{\mathrm{eff}}^{(m)}(r)$ denote the expected input state of one package, where the expectation here is only over the random shuffle, with the realized states $r_i^{(t)}$
held fixed.

We may label each distributed state by a tuple $\alpha=(i,t)\in \{1,\dots,n\}\times\{1,\dots,m\}$ so that we can write $r_\alpha=r_i^{(t)}$. A package consists of $n$ distinct labels
$\alpha_1,\dots,\alpha_n$, sampled without replacement from the $mn$
available labels. Therefore,
\begin{align}\label{eq:heterogeneous-fixed-realization}
    \rho_{\mathrm{eff}}^{(m)}(r) = \frac{1}{(mn)_n} \sum_{\substack{\alpha_1,\dots,\alpha_n\\ \text{all distinct}}} r_{\alpha_1}\otimes\dots\otimes r_{\alpha_n},
\end{align}
where $(mn)_n=mn(mn-1)\dots(mn-n+1)$ is the falling factorial. This is the direct analogue of the fixed-state formula, but it allows the $m$ outputs of a given source to be different. Note that the sum above is over ordered $(\alpha_1,\dots,\alpha_n)$. 

Now we average over the randomness of the sources. The total effective one-package state is
\begin{align}\label{eq:general-random-source-formula}
    \rho_{\mathrm{eff,tot}}^{(m)} = \E_R \left[\rho_{\mathrm{eff}}^{(m)}(R)\right] = \frac{1}{(mn)_n} \sum_{\substack{\alpha_1,\dots,\alpha_n\\ \text{all distinct}}} \E \left[ R_{\alpha_1}\otimes\dots\otimes R_{\alpha_n} \right].
\end{align}
Eq.~\eqref{eq:general-random-source-formula} is completely general. It shows that the effective state depends on joint moments of the random source outputs.

Suppose now that the random states $\{R_i^{(t)}:1\leq i\leq n,\ 1\leq t\leq m\}$ are mutually independent. Otherwise, we can also more weakly assume the tensor moment factorization $\E \left[ R_{\alpha_1}\otimes\dots\otimes R_{\alpha_n} \right] = \E[R_{\alpha_1}] \otimes\dots\otimes \E[R_{\alpha_n}]$. whenever $\alpha_1,\dots,\alpha_n$ are distinct, which is indeed what we will encounter under sampling without replacement. Then Eq.~\eqref{eq:general-random-source-formula} becomes
\begin{align}\label{eq:random-sources-replaced-by-means}
    \rho_{\mathrm{eff,tot}}^{(m)} &= \frac{1}{(mn)_n} \sum_{\substack{\alpha_1,\dots,\alpha_n\\ \text{all distinct}}} \bar{\rho}_{\alpha_1} \otimes \dots \otimes \bar{\rho}_{\alpha_n} \notag\\
    &= \rho_{\mathrm{eff}}^{(m)} \left( \{\bar{\rho}_i^{(t)}:1\leq i\leq n,\ 1\leq t\leq m\} \right).
\end{align}
Thus under independence assumption the random source outputs may be replaced by their mean density matrices before performing the shuffle average.

If the mean output of source $i$ is also the same in every round, i.e. $\bar{\rho}_i^{(1)} = \bar{\rho}_i^{(2)} = \dots = \bar{\rho}_i^{(m)} = \bar{\rho}_i$. then Eq.~\eqref{eq:random-sources-replaced-by-means} reduces exactly to the previous fixed-source formula.

\subsection{Application to a linear EPP map}
Consider the successful branch of the entanglement purification protocol as a linear, completely positive, trace-non-increasing (CPTNI) map $\mathcal{E}_{\mathrm{succ}}$. The un-normalized successful output of the EPP is linear in the input state, so we have
\begin{align}
    \rho_{\mathrm{succ}}^{\mathrm{un}} = \E \left[ \mathcal{E}_{\mathrm{succ}}(\rho_{\mathrm{package}}) \right] = \mathcal{E}_{\mathrm{succ}}(\E [\rho_{\mathrm{package}}]) = \mathcal{E}_{\mathrm{succ}}(\rho_{\mathrm{eff,tot}}^{(m)}),
\end{align}
where $\rho_{\mathrm{succ}}^{\mathrm{un}}$ denotes the unnormalized successful output state, $\rho_{\mathrm{package}}$ denotes the single-package state, and the expectation value is taken over both random sources and the random shuffle. The average success probability is
\begin{align}
    p_{\mathrm{succ}} = \Tr[\rho_{\mathrm{succ}}^{\mathrm{un}}] =  \Tr \left[\mathcal{E}_{\mathrm{succ}}(\rho_{\mathrm{eff,tot}}^{(m)}) \right].
\end{align}
Conditioned on success in the standard operational sense, the normalized average successful output state is
\begin{align}
    \rho_{\mathrm{succ}} = \frac{ \mathcal{E}_{\mathrm{succ}}(\rho_{\mathrm{eff,tot}}^{(m)}) }{ \Tr \left[ \mathcal{E}_{\mathrm{succ}}(\rho_{\mathrm{eff,tot}}^{(m)}) \right] },
\end{align}
provided that $p_{\mathrm{succ}}>0$.

\subsection{Including idling decoherence}
Idling decoherence can be incorporated in the same way. Suppose the raw state produced by source $i$ in round $t$ is $R_{i,\mathrm{raw}}^{(t)}$, and before the EPP it undergoes an idling decoherence channel $\mathcal E_i^{(t)}$. Then we have $R_i^{(t)} = \mathcal{E}_i^{(t)} \left(R_{i,\mathrm{raw}}^{(t)}\right)$.
By linearity of quantum channels we have
\begin{align}
    \bar{\rho}_i^{(t)} = \E[R_i^{(t)}] = \E\left[\mathcal{E}_i^{(t)} \left(R_{i,\mathrm{raw}}^{(t)}\right)\right] = \mathcal E_i^{(t)} \left( \E[R_{i,\mathrm{raw}}^{(t)}] \right).
\end{align}
Thus decoherence can be applied directly to the expected raw output state. For example, if the raw probability distribution is the same in every round of entanglement distribution, and the state from round $t$ undergoes $q_t$ identical idling steps where each idling step is described by $\mathcal{E}_i$, then we have
\begin{align}
    \bar{\rho}_i^{(t)} = \left(\mathcal{E}_i\right)^{q_t} \left(\bar{\rho}_i^{\mathrm{raw}}\right),
\end{align}
where $\left(\mathcal{E}_i\right)^{q_t}$ denotes the channel as a composition of $q_t$ single-step channel $\mathcal{E}_i$.

\end{document}